\documentclass[acmsmall,10pt,nonacm]{acmart}
\acmSubmissionID{124}

\AtBeginDocument{%
  \providecommand\BibTeX{{%
    \normalfont B\kern-0.5em{\scshape i\kern-0.25em b}\kern-0.8em\TeX}}}
\newtheorem{defn}{Definition}
\newtheorem{thm}{Theorem}
\newtheorem{prop}{Proposition}

\newtheorem{cor}{Corollary}

\newtheorem{lem}{Lemma}
\newtheorem{hyp}{Assumption}
\newtheorem{remark}{Remark}

\usepackage{subfigure}

\newcommand{\dE}{\mathbb{E}}

\newcommand{\cP}{\mathcal{P}}

\newcommand{\R}{\mathbb{R}}
\newcommand{\eps}{\varepsilon}
\newcommand{\F}{\mathcal{F}}
\newcommand{\N}{\mathbb{N}}	
\newcommand{\NRM}[1]{{{\left\| #1\right\|}}} 

\def\P{\mathbb P}
\def\E{\mathbb E}

\begin{document}

\title{\textbf{Asynchrony and Acceleration in Gossip Algorithms}}

\author{Mathieu Even}
\email{mathieu.even@inria.fr}
\affiliation{%
  \institution{Inria, ENS Paris, PSL Research University}
  \country{France}
}
\author{Hadrien Hendrikx}
\email{hadrien.hendrikx@inria.fr}
\affiliation{%
  \institution{Inria, ENS Paris, PSL Research University}
  \country{France}
}
\author{Laurent Massoulié}
\email{laurent.massoulie@inria.fr}
\affiliation{%
  \institution{Inria, ENS Paris, PSL Research University}
  \country{France}
}


\newcommand{\hadrien}[1]{{\color{blue} \bf [HH: #1]}}
\newcommand{\mathieu}[1]{{\color{red} \bf [HH: #1]}}

\begin{abstract}This paper considers the minimization of a sum of smooth and strongly convex functions dispatched over the nodes of a communication network. Previous works on the subject either focus on synchronous algorithms, which can be heavily slowed down by a few slow nodes (the \emph{straggler problem}), or consider a model of asynchronous operation (\citet{boyd2006gossip}) in which adjacent nodes communicate at the instants of  \emph{Poisson point processes}. We have two main contributions. 1) We propose \emph{CACDM} (a Continuously Accelerated Coordinate Dual Method), and for the Poisson model of asynchronous operation, we prove \emph{CACDM} to converge to optimality at an accelerated convergence rate in the sense of~\citet{neststich2017acdm}. In contrast, previously proposed asynchronous algorithms have not been proven to achieve such accelerated rate. While \emph{CACDM} is based on discrete updates, the proof of its convergence crucially depends on a continuous time analysis.
2) We introduce a new communication scheme based on \emph{Loss-Networks}, that is programmable in a fully asynchronous and decentralized way, unlike the Poisson model of asynchronous operation that does not capture essential aspects of asynchrony such as non-instantaneous communications and computations. Under this Loss-Network model of asynchrony, we establish for \emph{CDM} (a Coordinate Dual Method) a rate of convergence in terms of the eigengap of the Laplacian of the graph weighted by local effective delays. We believe this eigengap to be a fundamental bottleneck for convergence rates of asynchronous optimization. Finally, we verify empirically that \emph{CACDM} enjoys an accelerated convergence rate in the Loss-Network model of asynchrony.
\end{abstract}



\maketitle

\section{Introduction}

In this paper, we consider minimization of a function $f$ given by a sum of local functions:
\begin{equation}
    \min_{x\in\R^d} f(x):=\sum_{i=1}^n f_i(x).\label{eq:intro_pbm}
\end{equation}
A typical example is provided by Empirical Risk Minimization (ERM), in which the local functions $f_i$ correspond to the empirical risk evaluated on subsets of the whole dataset. We further assume that there is an underlying communication network, and that each $f_i$, or gradients thereof, can only be computed at node $i$ of this network. In the case of ERM, $f_i$ represents the empirical risk for the dataset available at node $i$. We aim to solve Problem~\eqref{eq:intro_pbm} in a decentralized fashion, where each node can only communicate with its neighbors in the graph. 

Another important example is that of network averaging. It corresponds to $f_i(x)=\NRM{x-c_i}^2$ where $c_i$ is a vector attached to node $i$. The solution of Problem~\eqref{eq:intro_pbm} is then provided by $x^\star=\frac{1}{n}\sum_{i=1}^nc_i$.

Typical decentralized approaches for this problem rely on gossip communications \citep{shah2009gossip} and first order local gradient steps~\citep{tang2018d2,zhang2013communicationefficientoptim,boyd2012distributed,scaman2017optimal,nedich2016achieving,sun2018mcgd}. Yet, these approaches often rely on global synchronous rounds, in which all nodes exchange with their neighbours at the same time. Such synchronous approaches are well suited to networks with homogeneous communication and computation delays. However the presence of a few slow links or nodes drastically degrades their performance. Our work targets asynchronous distributed algorithms, for which we aim to obtain fast rates of convergence in networks with heterogeneous computation and communication delays, while being competitive with synchronous approaches in homogeneous environments. 

\subsection{Main Contributions}

We consider local operations and communication schemes where each pair of neighbor nodes $(i,j)$ can exchange local variables
at  \emph{activation times} of the corresponding edge~$(ij)$. We denote by $\mathcal{P}_{ij}\subset \R^+$ the \emph{Point process} of the corresponding activation times. Upon activation of edge~$(ij)$, nodes~$i$ and~$j$ can exchange local variables such as gradients of their local functions and update their local variables accordingly. We mainly study two models for the point processes $\cP_{ij}$: \emph{i)} the Poisson model of asynchrony popularized by Boyd et al. \citep{boyd2006gossip} where $(\cP_{ij})_{(ij)\in E}$ are independent \emph{Poisson point processes} of rates $p_{ij}$ \citep{Klenke2014Ppp}. We refer to this model as the \emph{Poisson point process model} (\emph{P.p.p. model}).  \emph{ii)} A more complex model, inspired by loss networks (\citet{kelly1991lossnetwork}), that we call \emph{Refined Loss Network Model} (\emph{RLNM}), designed to capture essential aspects of asynchronous communications and computations. 

\subsubsection{Randomized Gossip and \emph{P.p.p. model}} We extend results obtained by \cite{boyd2006gossip} on gossip algorithms for network averaging to more general optimization problems of the form of Problem~\eqref{eq:intro_pbm} through a dual formulation. We obtain a convergence rate that depends on both the condition number of the optimization problem and the Laplacian matrix of the graph, weighted by the rates of the \emph{Poisson point processes} $\cP_{ij}$. The proof relies on a continuous-time analysis, which paves the way  for the introduction of an accelerated algorithm, \emph{CACDM} (\emph{Continuously Accelerated Coordinate Dual Method}). \emph{CACDM} can be interpreted as an accelerated coordinate gradient descent on the dual problem involving infinitesimal contractions. Using this interpretation we prove that \emph{CACDM} converges at an accelerated rate in the sense of~\citet{neststich2017acdm}. To the best of our knowledge, this is the first asynchronous algorithm proven to achieve accelerated convergence rates in the \emph{P.p.p. model}.

\subsubsection{Refined Loss Network Model} Though the \emph{P.p.p.} model is very convenient, it assumes that communications and computations are performed instantaneously. We thus modify the communication scheme in order to model communications in a more realistic way: busy nodes (\emph{i.e.}~computing or communicating nodes) are made unavailable for other nodes to communicate with. This model is directly inspired by Loss Networks, where \emph{busy} nodes are locked away from the network, which we refine by adding a \emph{busy-checking} operation. For this communication model, we derive a rate of convergence that depends on the Laplacian matrix of the graph weighted by local communication constraints. Thus, we are able to recover the robustness to stragglers that we had with the P.p.p. model, but with a theory that is more faithful to the implementation. The construction and analysis of this model enable us to identify key parameters of the communication network that condition achievable convergence rates for realistic asynchronous and distributed operation.

\subsection{Related Work}

\subsubsection{Gossip Algorithms and Asynchrony}

In gossip averaging algorithms \citep{boyd2006gossip,dimakis2010synchgossip}, nodes of the network communicate with their neighbors without any central coordinator in order to compute the global average of local vectors. These algorithms are particularly relevant since they can be generalized to address our distributed optimization problem with local functions $f_i$ beyond the special case $f_i(x)=\NRM{x-c_i}^2$. Two types of gossip algorithms appear in the literature: synchronous ones, where all nodes communicate with each other simultaneously \citep{scaman2017optimal,dimakis2010synchgossip,berthier2018acceleratedgossip}, and asynchronous ones also called randomized gossip
\citep{boyd2006gossip,nedic2009ieee,hendrikx2018accelerated}, where at a defined time~$t\geq0$, only a  pair of adjacent nodes can communicate. In the synchronous framework, the communication speed is limited by the slowest node (\emph{straggler} problem).

Although qualified as asynchronous, the  \emph{P.p.p.} model cannot be programmed in a fully distributed and asynchronous structure: it assumes that communications and computations are instantaneous. Two different approaches can be considered to deal with the fact that communications and computations are in fact non-instantaneous: (i) when a node $i$ receives information from a neighbor $j$ at a time $t\geq 0$, account for the fact that this information is delayed, or (ii) forbid communications with a \emph{busy} (\emph{i.e.} communicating or computing) edge, thereby removing the need to handle delayed information. The first approach (i) is considered for asynchronous but centralized optimization by \citep{leblond2016asaga,niu2011hogwild}, where delayed variables are modelled as so-called \emph{perturbed iterates}. The second approach (ii) is reminiscent of \emph{Loss-Networks}, initially considered for telecommunication networks \citep{kelly1991lossnetwork}, yet also adequate to reflect primitives in distributed computing such as \emph{locks} and \emph{atomic transactions}. 

In the \emph{perturbed iterate} modelling, a central unit delegates computations to workers. Asynchrony lies in the fact that these workers do not wait for the central unit to update their current version of the optimization variable $x$, but instead  send gradients $\nabla f_i(x_i)$ whenever they can, even if based on outdated variable $x_i$. Thus, the parameter of the central unit is updated using perturbed (\emph{delayed}) gradients \citep{mania2015perturbed}. Section \ref{section:LN} focuses on the second modelling: nodes behave as in the $\emph{P.p.p.~model}$, but are made \emph{busy} and hence non-available for other nodes for a time $\tau_{ij}>0$ after their activation. The system is asynchronous in the sense that communications are performed in a random pairwise fashion (instead of global synchronous rounds), and nodes do not wait for specific neighbours. Yet, received gradients are never out of date since nodes always finish their current operation (communicating or computing) before engaging in a new one.

\subsubsection{Acceleration in an Asynchronous Setting}
Acceleration means gaining order of magnitudes in terms of convergence speed, compared to classical algorithms. Accelerating gossip algorithms has been studied in previous works in the synchronous framework: \emph{SSDA} \citep{scaman2017optimal}, Chebyshev acceleration \citep{montijano2011chebgossip}  Jacobi-Polynomial acceleration in the first iterations \citep{berthier2018acceleratedgossip}, or in the asynchronous \emph{P.p.p.~model}: Geographic Gossip \citep{Dimakis_2008} , shift registers \citep{LIU2013873}. However, no algorithm in the \emph{P.p.p.~model} has been rigorously proven to achieve an  accelerated rate for general graphs without additional synchronization between nodes. For instance, inspired by \emph{ACDM} \citep{neststich2017acdm}, \cite{hendrikx2018accelerated} introduced \emph{ESDACD}, where at each iteration, only a pair of adjacent nodes communicate, but all nodes need to make local contractions and thus need to know that an update is taking place somewhere else in the graph. This last requirement, also present in \emph{Stochastic Heavy Balls} methods \citep{loizou2018accelerated}, is a departure from purely asynchronous operation, and thus a limitation of these methods.
Section \ref{section:cacdm} presents a continuous alternative to \emph{ACDM}, where the contractions previously cited are made continuously. Our algorithm (\emph{CACDM}, for Continuously Accelerated Coordinate Descent Method) obtains in the \emph{P.p.p.~model} the same accelerated rate as \cite{Dimakis_2008,loizou2018accelerated,hendrikx2018accelerated} for any graph, without assuming access to any global iteration counter: it only needs local clock-synchronization between adjacent nodes.   Although our analysis of \emph{CACDM} does not extend to more general communication models such as those presented in Section \ref{section:LN}, we observe empirically that \emph{CACDM} enjoys accelerated rates in the Loss-Network model as well as in the P.p.p. model.\\

The detailed problem statement and notations are given in Section \ref{section:pbm}. Section~\ref{section:ppp} contains our results on asynchronous gossip in the \emph{P.p.p.}~model, first for a non-accelerated algorithm based on simple gradient descent steps, then for the accelerated algorithm \emph{CACDM}. Section~\ref{section:LN} finally presents our results for gossip algorithms in  the \emph{refined loss network model}.

\section{Problem Formulation and Notations \label{section:pbm}}

\subsection{Basic assumptions and notations\label{section:formul_def}}

The communication network is represented by an undirected graph $G=(V,E)$ on the set of nodes $V=[n]$, and is assumed to be connected. Two nodes are said to be neighbors or adjacent  in the graph, and we write $i\sim j$, if $(ij)\in E$. Two edges $(ij),(kl)\in E$ are adjacent in the graph if $(ij)=(kl)$ or if they share a node. Each node $i\in V$ has access to a local function $f_i$ defined on $\R^d$, assumed to be $L_i$-smooth and $\sigma_i$-strongly convex \citep{bubeck2014convex}, \emph{i.e.} $\forall x,y\in \R^d$:
\begin{equation}\label{eq:smooth_sc}
\begin{split}
    & f_i(x)\leq f_i(y)+\langle\nabla f_i(y),x-y\rangle +\frac{L_i}{2}\NRM{x-y}^2, \\
    & f_i(x)\geq f_i(y)+\langle\nabla f_i(y),x-y\rangle +\frac{\sigma_i}{2}\NRM{x-y}^2.
    \end{split}
\end{equation}
Let us denote $f(z)=\sum_{i\in [n]} f_i(z)$ for $z\in \R^d$ and  $F(x)=\sum_{i\in [n]} f_i(x_i)$ for $x=(x_1^\top,\cdots,x_n^\top)\in \R^{n\times d}$ 
where $x_i\in \R^d$ is attached to node $i\in[n]$. Let
\begin{equation}\label{eq:sigma_min_L_max}
        L_{\max}:=\max_i L_i \text{ and }\sigma_{\min}:=\min_i \sigma_i
\end{equation} denote the global complexity numbers. 
Computing gradients and communicating them between two neighboring nodes $i\sim j$ is assumed to take time $\tau_{ij}>0$. This constant takes into account both the communication and computation times, and should be understood as an upper-bound on the delays between nodes $i$ and $j$. 

\noindent In this decentralized setting, Problem~\eqref{eq:intro_pbm} can be formulated as follows:
\begin{equation}
    \label{eq:main_problem_primal}
    \min_{x\in \R^{n\times d}:x_1=...=x_n} F(x),
\end{equation}
where $x_1=...=x_n$ enforces consensus on all the nodes. We add the following structural constraints:
\begin{enumerate}
    \item \emph{ Local computations:} node $i$ (and node $i$ only) can compute first-order characteristics of $f_i$ such as $\nabla f_i$ or $\nabla f_i^*$;
    \item \emph{ Local communications:} node $i$ can send information only to neighboring nodes $j\sim i$.
\end{enumerate}

\noindent These operations may be performed asynchronously and in parallel, and each node possesses a local version $x_i\in \R^d$ of the global parameter $x$. The rate of convergence of our algorithms will be controlled by the smallest positive eigenvalue $\gamma$ of the Laplacian of graph~$G$ \citep{Mohar1991laplacian}, weighted by some constants $\nu_{ij}$ that depend on the local communication and computation delays. 
\begin{defn}[Graph Laplacian]\label{laplacian}
Let $(\nu_{ij})_{(ij)\in E}$ be a set of non-negative real numbers. The Laplacian of the graph $G$ weighted by the $\nu_{ij}$'s is the matrix with $(i,j)$ entry equal to $-\nu_{ij}$ if $(ij)\in E$, $\sum_{k\sim i} \nu_{ik}$ if $j=i$, and $0$ otherwise. In the sequel $\nu_{ij}$ always refers to the weights of the Laplacian, and $\gamma(\nu_{ij})$ denotes this Laplacian's second smallest eigenvalue.
\end{defn}
\noindent For any function $g:\R^p\to\R$, $g^*$ denotes its \emph{Fenchel conjugate} on $\R^p$ defined as
\begin{equation*}
    \forall y\in\R^p, g^*(y)=\sup_{x\in \R^p}\langle x,y\rangle -g(x)\in \R\cup \{+\infty\}.
\end{equation*}
Throughout the paper, $\F_t$ for $t\in \R^+$ denotes the filtration of the point processes $\cP=\bigcup_{(ij)\in E}\cP_{ij}$ up to time $t$. If $t_k,k\in \N^*$ (and $t_0=0$) are the successive points in $\cP$, we write if there is no ambiguity $\F_k=\F_{t_k},k\in \N^*$.

\subsection{Dual Formulation of the Problem\label{section:dual}}

A standard way to deal with the constraint $x_1=...=x_n$, is to use a dual formulation \citep{scaman2017optimal,hendrikx2018accelerated,uribe2020dual}, by introducing a dual variable $\lambda$ indexed by the edges. We first introduce a matrix $A\in \R^{n \times E}$ such that $\rm Ker (A^\top)=Vect(\mathbb{I})$ where $\mathbb{I}$ is the constant vector $(1,...,1)^\top$ of dimension $n$. $A$ is chosen such that:
\begin{equation}\label{eq:matrixA}
    \forall (ij)\in E, A e_{ij} = \mu_{ij} (e_i-e_j).
\end{equation} for some non-null constants $\mu_{ij}$. We define $\mu_{ij}=-\mu_{ji}$ for this writing to be consistent. This matrix $A$ is a square root of the laplacian of the graph weighted by $\nu_{ij}=\mu_{ij}^2$.
The constraint $x_1=...=x_n$ can then be written $A^\top x=0$. The dual problem reads as follows:
\begin{align*}
    &\min_{x \in \R^{n\times d},A^\top x=0} \sum_{i=1}^n f_i(x_i) =\min_{x \in \R^{n\times d}} \max_{\lambda \in \R^E} \sum_{i=1}^n f_i(x_i) -\langle A^\top x,\lambda\rangle.
\end{align*} 
Let $F_A^*(\lambda):=F^*(A\lambda)$ for $\lambda \in \R^{E\times d}$ where $F^*$ is the Fenchel conjugate of $F$. The dual problem reads
\begin{equation*}
    \min_{x \in \R^{n\times d},x_1=...=x_n} F(x) =\max_{\lambda\in \R^{E \times d}} -F_A^*(\lambda).
\end{equation*}
Thus $F_A^*(\lambda)=\sum_{i=1}^n f_i^*((A\lambda)_i)$ is to be minimized over the dual variable $\lambda \in \R^{E\times d}$. 

We now make a parallel between pairwise operations between adjacent nodes in the network and coordinate gradient steps on $F_A^*$. As $F_A^*(\lambda)=\max_{x\in \R^{n\times d}} -F(x)+\langle A\lambda,x\rangle$, to any $\lambda\in \R^{E\times d}$ a primal variable $x\in \R^{n\times d}$ is uniquely associated through the formula $\nabla F(x)=A\lambda$.
The partial derivative of $F^*_A$ with respect to coordinate $(ij)$ of $\lambda$ reads :
\begin{align*}
    \nabla_{ij} F_A^*(\lambda)&=(A e_{ij})^\top \nabla F^*(A\lambda)=\mu_{ij}(\nabla f_i^*((A\lambda)_i)- \nabla f_j^*((A\lambda)_j)).
\end{align*}
Consider then the following step of coordinate gradient descent for $F^*_A$ on coordinate $(ij)$ of $\lambda$, performed when edge $(ij)$ is activated at iteration $k$ (corresponding to time $t_k$), and where $U_{ij}=e_{ij}e_{ij}^\top$:
\begin{equation}\label{eq:step_lambda}
    \lambda_{t_{k+1}}=\lambda_{t_k}-\frac{1}{(\sigma_i^{-1}+\sigma_j^{-1})\mu_{ij}^2}U_{ij}\nabla_{ij} F_A^*(\lambda_{t_k}).
\end{equation}
Denoting $v_k=A\lambda_{t_k}\in \R^{n\times d}$, we obtain the following formula for updating coordinates $i,j$ of $v$ when $ij$ activated:
\begin{align}
    & v_{k+1,i}=v_{k,i}-\frac{\nabla f_i^*(v_{k,i})-\nabla f_j^*(v_{k,j})}{\sigma_i^{-1}+\sigma_j^{-1}}\label{eq:v_1},\\
    & v_{k+1,j}=v_{k,j}+\frac{\nabla f_i^*(v_{k,i})-\nabla f_j^*(v_{k,j})}{\sigma_i^{-1}+\sigma_j^{-1}}\label{eq:v_2}.
\end{align}
Such updates can be performed locally at nodes $i$ and $j$ after communication between the two nodes. We refer in the sequel to this scheme as the Coordinate Descent Method (CDM).
While $\lambda\in\R^{E\times d}$ is a dual variable defined on the edges, $v\in \R^{n\times d}$ is also a dual variable, but defined on the nodes. The {\em primal surrogate} of $v$ is defined as $x=\nabla F^*(v)$ \emph{i.e.}~$x_i=\nabla f_i^*(v_i)$ at node $i$. It can hence be computed with local updates on $v$ (\eqref{eq:v_1} and \eqref{eq:v_2}).
Thus \emph{CDM}, based on coordinate gradient descent for the dual problem, translates into \emph{local} updates for the primal variables $x_i$. Note that in order to perform \emph{CDM}, an initialization $v(0)\in {\rm Im}(A)$ at all nodes is required, to ensure the existence of $\lambda\in \R^{E\times d}$ such that $A\lambda(0)=v(0)$. We thus usually take $v_i(0)=0$ for all nodes $i$.

\begin{remark}\label{remark:conjugates}
We hence have two notions of duality. For $x=(x_1,...,x_n)\in \R^{n\times d}$ the primal variables associated with the network nodes, $v=(v_1,...,v_n)\in\R^{n\times d}$ is its convex-dual conjugate with $v_i=\nabla f_i(x_i)$, while $\lambda\in\R^{E\times d}$ such that $A\lambda=v$ is its \emph{edge-dual conjugate}.
\end{remark}

\begin{remark}
Matrix $A$ is introduced only for the purpose of the analysis. Indeed, we analyze our algorithms through edge-dual formulations, with updates of the form \eqref{eq:step_lambda} on these variables. However, we present the algorithm with the convex-dual variables, \eqref{eq:v_1},\eqref{eq:v_2}, for which $\mu_{ij}^2$ and hence the effect of matrix~$A$ disappears.
\end{remark}

\subsection{Gossip Averaging Problem}

As previously mentioned, the initial problem \eqref{eq:intro_pbm} with functions $f_i(x)=\frac{1}{2}\NRM{x-c_i}^2,x\in \R^d$ for some vectors $c_1,...,c_n\in \R^d$ reduces to the gossip averaging problem that aims at computing in a decentralized way with local computations the value $\Bar{c}=\frac{1}{n}\sum_{i=1}^nc_i$. We contrast in this particular framework the rates that can be obtained by synchronous and asynchronous methods. These rates are expressed in terms of the weighted graph Laplacian, where for synchronous updates the edge weights are tuned to the worst-case delay, whereas in the asynchronous case, the edge weights can be tuned to local delay. Thus the advantage of asynchronous methods over synchronous ones is captured by these different edge weights in the considered Laplacian.


\noindent \textbf{Synchronous Communications:} 
In Synchronous Gossip Algorithm iterations \citep{dimakis2010synchgossip}, all nodes update their values synchronously by taking a weighted average of the values of their neighbors (Appendix \ref{app:synch_gossip} for more details). These algorithms converge linearly with a rate given by the smallest eigenvalue of the graph Laplacian  weighted by weights $\nu_{ij}\leq 1$. Since every iteration takes a time $\tau_{max}$, synchronous Gossip algorithms have a linear rate of convergence $\gamma_{synch}=\gamma(\nu_{ij})$ with weights $\nu_{ij}\leq\tau_{max}^{-1}$ for all $(ij)\in E$ (Definition \ref{laplacian}). We rephrase this as the following

\begin{prop}[Synchronous Gossip] Let $x(t)=(x_1(t),...,x_n(t))^\top\in \R^{n\times d}$ be the matrix of vectors $x_i(t)$ attached to node $i$ at time $t\ge0$. For continuous time $t\ge 0$ and for synchronous gossip algorithms as in \cite{dimakis2010synchgossip}, we have:
\begin{equation}\label{eq:synch_gossip}
    \NRM{x(t)-\Bar{c}}^2\leq \exp(-(t-\tau_{\max})\gamma_{synch})\NRM{x(0)-\Bar{c}}^2,
\end{equation}
with $\gamma_{synch}$ the second smallest eigenvalue of the graph Laplacian weighted by $\nu_{ij}\equiv\tau_{\max}^{-1}$ .
\end{prop}

\noindent \textbf{Asynchronous Communications in the \emph{P.p.p. model}:} This is the setting of randomized gossip as considered by \cite{boyd2006gossip}, where point processes $\cP_{ij}$ are independent \emph{P.p.p.} of rates $p_{ij}>0$. When edge $(ij)$ is activated, nodes $i$ and $j$ update their values by making a local averaging (Appendix \ref{app:asynch_gossip}). We have the following convergence result.
\begin{prop}[Randomized Gossip] For randomized gossip as in \cite{boyd2006gossip}, we have:
\begin{equation}\label{eq:asynch_gossip}
\E[\NRM{x(t)-\Bar{c}}^2]\leq \exp(-t\gamma_{asynch})\NRM{x(0)-\Bar{c}}^2,
\end{equation}
with $\gamma_{asynch}$ the second smallest eigenvalue of the graph Laplacian weighted by $\nu_{ij}=p_{ij}$.
Moreover, this rate is optimal in the sense that there exists $x(0)\in\R^{n\times d}$ such that \eqref{eq:asynch_gossip} is an equality for all $t\ge 0$.
\end{prop}

Proofs of \eqref{eq:synch_gossip} and \eqref{eq:asynch_gossip} and details about synchronous and randomized gossip can be found in Appendix \ref{app:gossip}. Equation~\eqref{eq:asynch_gossip} follows from derivations in \cite{boyd2006gossip}, combined with a study of infinitesimal intervals of times $[t,t+dt]$. We generalize this result to the initial optimization problem \eqref{eq:intro_pbm} in next Section.

In the P.p.p., the terms $1/p_{ij}$ capture the average time between consecutive activations of edge $(ij)$ and are thus naturally related to the delays $\tau_{ij}$. This suggests that asynchrony brings about a speed-up reflected by the change in the Laplacian's spectral gap $\gamma(\nu_{ij})$ when the weights $\nu_{ij} \equiv\tau_{\max}^{-1}$ are replaced by $\nu_{ij}=\tau_{ij}^{-1}$.
The fact that $\gamma_{asynch}$ is optimal leads us to believe that this quantity - the smallest non-null eigenvalue of the Laplacian with local weights - best describes the asynchronous speed-up. 

The above argument identifying $\tau_{ij}$ with $p_{ij}^{-1}$ is heuristic. Our analysis of the Loss-Network model will establish a more  rigorous bridge between spectral gap of Laplacian with edge weights based on local delays and convergence speed of asynchronous schemes.

\section{Randomized Gossip: the \emph{P.p.p. model}\label{section:ppp}}

\subsection{The \emph{P.p.p. Model} and Randomized Gossip Algorithms}

\textbf{The \emph{P.p.p. model}:} Each edge $(ij)\in E$ has a clock that ticks at the instants of a  \emph{Poisson point process} $\mathcal{P}_{ij}$ of intensity $p_{ij}$, where the $\mathcal{P}_{ij}$ are mutually  independent. At each tick of its clock,  edge $(ij)$ is activated and nodes $i$ and $j$ can communicate together. The process $\mathcal{P}=\bigcup_{(ij)\in E}\mathcal{P}_{ij}$, $\mathcal{P}$ is again \emph{P. p. p.} of intensity \begin{equation}
    I=\sum_{(ij)\in E} p_{ij}.\label{eq:def_I}
\end{equation}

\noindent \textbf{Randomized Gossip Algorithm:} Each node $i$ maintains a local variable $(x_i(t))_{t\ge0}$. We denote $(v_i(t))_{t\ge0}$ its local convex-dual conjugate and write $v(t)=(v_i(t))_i$. We initialize with $v_i(0)=0$ at all nodes. Based on the dual problem formulation in Section \ref{section:dual}, we consider \emph{CDM}. Specifically,  when clock $(ij)$ ticks at time $t\ge0$, perform the following update on variable $v(t)$:
\begin{align}
\begin{split}
    & v_{i}(t)\xleftarrow t v_{i}(t)-\frac{\nabla f_i^*(v_{i}(t))-\nabla f_j^*(v_{j}(t))}{\sigma_i^{-1}+\sigma_j^{-1}},\\
    & v_{j}(t)\xleftarrow t v_{j}(t)+\frac{\nabla f_i^*(v_{i}(t))-\nabla f_j^*(v_{j}(t))}{\sigma_i^{-1}+\sigma_j^{-1}}.
    \end{split}\label{eq:v_updates}
\end{align}
The desired output at node $i$ and time $t$ is then $x_i(t)=\nabla f_i^*(v_i(t))$. Note that as mentioned in Section \ref{section:dual}, the outputs $v_i(t)$ and $x_i(t)$ at any node $i$ and time $t$ are all completely independent from the initial choice of matrix $A$, whose only use is for analysis. Observe that in the gossip averaging problem, $v_i(t)=x_i(t)-x_i(0)$, and Equation~\eqref{eq:v_updates} simplifies to
\begin{equation}
    x_i(t),x_j(t) \xleftarrow t \frac{x_i(t)+x_j(t)}{2},\label{eq:gossip_update}
\end{equation} which coincides with classical randomized gossip updates for the averaging problem.

\subsection{Continuous Time Convergence Analysis}

The classical analysis of gossip algorithms \citep{boyd2006gossip} proceeds as follows: at every clock tick of $\mathcal{P}$, an edge $(ij)$ is selected with probability $q_{ij}=\frac{p_{ij}}{I}$. 
A discrete time analysis of state variables at these ticking times is then performed. 
In order to derive bounds for continuous time $t$, we instead study infinitesimal intervals of time $[t,t+dt]$, giving us more degrees of freedom, as shown in Section \ref{section:cacdm}.
\begin{thm} \label{thm:ppp_standard} For the \emph{CDM} updates \eqref{eq:v_updates}, in the P.p.p. model with intensities $p_{ij}$, we have the following guarantees for all $t\ge 0$
\begin{equation}
    \dE(F^*(v(t))-F^*(v^\star))\le (F^*(v(0))-F^*(v^\star)) \exp\left(-\frac{\sigma_{\min}}{2L_{\max}}\gamma_p t\right),
\end{equation}
where $v^\star=A \lambda^\star $ is the minimizer of $F^*$ on ${\rm Im}(A)$, $\lambda^\star $ being a minimizer of $F^*_A$, $\gamma_p=\gamma(p_{ij})$ is the spectral gap of the graph Laplacian  weighted by weights $\nu_{ij}=p_{ij}$ 
 and $\sigma_{\min},L_{\max}$ are defined in \eqref{eq:sigma_min_L_max}.
\end{thm}
Since $x(t)=\nabla F^*(v(t))$ and $x^\star=\nabla F^*(v^\star)$ where $x^\star$ is the minimizer of $F$ under the consensus constraint, we have on primal variable $x(t)$ (Lemma \ref{lem:primal_dual}):
\begin{equation}
    \E\left[\NRM{x_t-x^\star}^2\right]\le \frac{2L_{\max}}{\sigma_{\min}^2}(F^*(v(0))-F^*(v^\star)) \exp\left(-\frac{\sigma_{\min}}{2L_{\max}}\gamma_p t\right).
\end{equation}
We thus obtain a factor $\gamma_p$ in the rate of convergence that reflects communication speed, and $\frac{\sigma_{\min}}{L_{\max}}$ that is an upper-bound on the condition number of the objective function. The sketch of proof below relies on a classical analysis of coordinate descent algorithms adapted to continuous time. The technical details are differed to Appendix \ref{app:prelim}. We believe the proof technique to be of independent interest: it could be applied to analyze optimization methods such as gradient descent algorithms (stochastic, proximal or accelerated ones) with increments ruled by \emph{Poisson point processes} with simple proofs based on establishment of differential inequalities. 
\begin{proof} We prove Theorem \ref{thm:ppp_standard} by considering edge-dual variables $\lambda_t\in\R^{E\times d}$ associated to $x(t)$ and $v(t)$, in particular with $A\lambda_t=v(t)$ and $A\lambda^\star =v^\star$. Since $v(0)=0$, we take $\lambda_0=0$. We consider matrix $A$ in \eqref{eq:matrixA} with $\mu_{ij}^2=\frac{p_{ij}}{\sigma_i^{-1}+\sigma_j^{-1}}$.  When clock $(ij)$ ticks at time $t\ge0$, the following update is performed on variable $\lambda_t$:
\begin{equation}\label{eq:thm1}
    \lambda_{t} \xleftarrow t \lambda_{t}-\frac{1}{(\sigma_i^{-1}+\sigma_j^{-1})\mu_{ij}^2}U_{ij}\nabla_{ij} F_A^*(\lambda_{t}).
\end{equation}
Furthermore, note that we have $F^*(v(t))=F_A^*(\lambda_t).$
A key ingredient in the proof is the lemma below,  which establishes a local smoothness property. Its proof is given in Appendix \ref{app:prelim},
\begin{lem}\label{lemma:smoothness} For $\lambda\in \R^{E\times d}$ and $ij\in E$, we have:
\begin{equation}
    F_A^*\left(\lambda-\frac{1}{\mu_{ij}^2(\sigma_i^{-1}+\sigma_j^{-1})}U_{ij}\nabla_{ij} F_A^*(\lambda)\right)-F_A^*(\lambda)\leq -\frac{1}{2\mu_{ij}^2(\sigma_i^{-1}+\sigma_j^{-1})}\NRM{\nabla_{ij}F_A^*(\lambda)}^2.
    \label{smoothness}
\end{equation} 
\end{lem}
\noindent Then using this, for $t\ge 0$ and $dt>0$:
\begin{align*}
    \E^{\F_t}[F_A^*(\lambda_{t+dt})-F_A^*(\lambda_{t})] & = (1-Idt)\E^{\F_t}[F_A^*(\lambda_{t+dt})-F_A^*(\lambda_{t})|\text{no activations in }[t,t+dt]]\\
    &+\sum_{(ij)\in E} p_{ij}dt\E^{\F_t}[F_A^*(\lambda_{t+dt})-F_A^*(\lambda_{t})|(ij)\text{ activated in }[t,t+dt]]\\
    &=-dt\sum_{ij \in E} p_{ij} (F_A^*(\lambda_t)-F_A^*(\lambda_t-\frac{1}{(\sigma_i^{-1}+\sigma_j^{-1})\mu_{ij}^2}U_{ij}\nabla_{ij}F_A^*(\lambda_t))) + o(dt)\\
    &\le -dt\sum_{ij\in E} \frac{p_{ij}}{2(\sigma_i^{-1}+\sigma_j^{-1})\mu_{ij}^2}\NRM{\nabla_{ij}F_A^*(\lambda_t)}^2+o(dt)\\
    &= - \frac{dt}{2}\NRM{\nabla F_A^*(\lambda_t)}^2+o(dt)
\end{align*}
Lemma \ref{lemma:grad_domination} in the Appendix implies that $\NRM{\nabla F_A^*(\lambda)}\ge 2\sigma_A (F_A^*(\lambda)-F_A^*(\lambda^\star ))$, where $\sigma_A$ is the strong convexity parameter of $F_A^*$ with respect to the Euclidean norm on the orthogonal of ${\rm Ker}(A)$. We thus have:
\begin{align*}
    \E^{\F_t}[F_A^*(\lambda_{t+dt})-F_A^*(\lambda_{t})]\le -dt\sigma_A (F_A^*(\lambda_t)-F_A^*(\lambda^\star ))+o(dt).
\end{align*}
Then, dividing by $dt$ and taking $dt\to 0$ yields: $\frac{d}{dt}\E[F_A^*(\lambda_t)-F_A^*(\lambda^\star )]\le -\sigma_A\E[F_A^*(\lambda_t)-F_A^*(\lambda^\star )]$. We then obtain an exponential rate of convergence $\sigma_A$ by integrating. Finally, Lemma \ref{lemma:sc} in the Appendix gives $\sigma_A\ge \frac{\lambda^+_{\min}(AA^\top)}{L_{max}}$  where $\lambda^+_{\min}(AA^\top)$ is the smallest non-null eigenvalue of $AA^\top$. As $AA^\top$ is the Laplacian of the graph with weights $\nu_{ij}=\mu_{ij}^2=\frac{p_{ij}}{\sigma_i^{-1}+\sigma_j^{-1}}$ (Lemma \ref{lemma:AAT}), we have $\lambda^+_{\min}(AA^\top)\ge \sigma_{\min}\gamma_p/2$ and \eqref{eq:thm1} follows.
\end{proof}

\begin{remark} The above study of infinitesimal intervals of time directly leads to continuous-time bounds. These could also be derived from a discrete time analysis:  Denote by $t_k\ge0$ the time of $k$-th activation, $k\in\N^*$, and $t_0=0$. We can prove that:
\begin{equation}
    \E[F_A^*(\lambda_{t_k})-F_A^*(\lambda^\star )]\le (1-\sigma_A/I)^k (F_A^*(\lambda_{0})-F_A^*(\lambda^\star )),
\end{equation}
where $I=\sum_{(ij)\in E}p_{ij}$. Then, we have in continuous time, for any $t\in\R^+$:
\begin{align*}
    \E[F_A^*(\lambda_t)-F_A^*(\lambda^\star )] & = \sum_{k\in \N} \frac{e^{-It}(It)^k}{k!} \E[F_A^*(\lambda_t)-F_A^*(\lambda^\star )|k \text{ activations in }[0,t]]\\
    &\le \sum_{k\in \N} \frac{e^{-It}(It)^k}{k!} (1-\sigma_A/I)^k (F_A^*(\lambda_0)-F_A^*(\lambda^\star ))\\
    &= e^{-\sigma_A t}(F_A^*(\lambda_0)-F_A^*(\lambda^\star )),
\end{align*}
giving the same result. However in the next Section, we will see that the continuous time viewpoint is essential in the design of the  \emph{CACDM} algorithm, as well as for its analysis through consideration of infinitesimal intervals and differential calculus.
\end{remark}

\subsection{Accelerated Gossip in the \emph{P.p.p. model} \label{section:cacdm}}

Inspired by previous works \citep{neststich2017acdm,hendrikx2018accelerated}, we propose \emph{CACDM} (Continuously Accelerated Coordinate Descent Method),  a gossip algorithm that, for the \emph{P.p.p. model},  provably obtains an accelerated rate of convergence in the sense of \citet{neststich2017acdm} (Theorem \ref{thm:cacdm}).

\subsubsection{\emph{CACDM} algorithm and convergence guarantees}

Similarly to other standard acceleration techniques, the algorithm requires maintaining two variables $x(t),y(t)\in \R^{n\times d}$, whose convex-dual conjugates are denoted $u(t),v(t)$. The variable $v(t)$ plays the role of a momentum. 
We initialize such that $u(0)=v(0)=0$. 
As in last subsection, variables $u_i(t),v_i(t),x_i(t),y_i(t)\in \R^d$ for $i\in[n]$ are attached to node $i$. The algorithm involves two types of operations: continuous contractions, and pairwise updates along each edge $(ij)$ when its Poisson clock ticks. 
\begin{enumerate}
    \item \textbf{Continuous Contractions:} For all times $t\in \R^+$ and node $i\in[n]$, for some fixed $\theta>0$ to be specified, make the infinitesimal contraction \begin{equation*}
    \begin{pmatrix} u_i(t+dt) \\ v_i(t+dt) \end{pmatrix}=\begin{pmatrix} 1- dt I\theta & dt I\theta \\ dt I\theta & 1-dt I\theta \end{pmatrix}\begin{pmatrix} u_i(t) \\ v_i(t) \end{pmatrix},\end{equation*}  between times $t$ and $t+dt$. Between times $s<t$, if there is no activation of $i$, it consists in performing the contraction: \begin{align}\label{eq:contract}
\begin{pmatrix} u_i(t) \\ v_i(t) \end{pmatrix}&=\exp{\left((t-s)I\begin{pmatrix} -\theta & \theta \\ \theta & -\theta \end{pmatrix}\right)}\begin{pmatrix} u_{i}(s) \\ v_{i}(s) \end{pmatrix}=\begin{pmatrix} \frac{1+e^{-2I\theta (t-s)}}{2} & \frac{1-e^{-2I\theta (t-s)}}{2} \\ \frac{1-e^{-2I\theta (t-s)}}{2} & \frac{1+e^{-2I\theta (t-s)}}{2} \end{pmatrix}\begin{pmatrix} u_{i}(s) \\ v_{i}(s) \end{pmatrix}.
\end{align}\\
    \item \textbf{Local Updates:} Let $\gamma_p$ be the smallest non-null eigenvalue of the Laplacian of the graph weighted by the local rates: $\nu_{ij}=p_{ij}$ (Definition \ref{laplacian}), and $L_{\max}$ defined in \eqref{eq:sigma_min_L_max}. When edge $(ij)$ is activated at time $t\geq 0$, perform the local update between nodes $i$ and $j$:
\begin{align}
      &u_i(t)\xleftarrow t u_i(t)-\frac{\nabla f_i^*(u_t(i))-\nabla f_j^*(u_t(j))}{\sigma_i^{-1}+\sigma_j^{-1}},\label{eq:x_step}\\
        &v_i(t)\xleftarrow t v_i(t)-\frac{\theta L_{\max}}{\gamma_p}\left(\nabla f_i^*(u_t(i))-\nabla f_j^*(u_t(j)\right),\label{eq:y_step}     
    \end{align}
and symmetrically at node $j$. The desired output at node $i$ and at time $t$ is then $x_i(t)=\nabla f_i^*(u_{i}(t))$ (Section \ref{section:dual}).
\end{enumerate}
\noindent \noindent This procedure can be performed asynchronously and at discrete times: the length $t-s$ between two activations of an edge that appears in the exponential contraction \eqref{eq:contract} is a local variable that can be computed from a  local clock. More formally, the stochastic process defined above is the following, where $V_t=(u(t),v(t))^T$ and $N_{ij}$ are independent \emph{P.p.p.} of intensities $p_{ij}$:
\begin{equation*}
dV_t=I\begin{pmatrix} -\theta & \theta \\ \theta & -\theta \end{pmatrix} V_t dt - \sum_{(ij)\in E}dN_{ij}(t)\begin{pmatrix}
    \frac{\nabla f_i^*(u_t(i))-\nabla f_j^*(u_t(j))}{\sigma_i^{-1}+\sigma_j^{-1}}\\
    \frac{\theta L_{\max}}{\gamma_p}\left(\nabla f_i^*(u_t(i))-\nabla f_j^*(u_t(j)\right)
\end{pmatrix}.
\end{equation*}

\noindent Define the Lyapunov function
\begin{equation}\label{eq:lyap1}
    \mathcal{L}_t={\NRM{v(t)-v^\star}^2_{({A^*}^\top A^*)^2}}+\frac{2\theta^2S^2 L_{\max}^2}{\gamma_p^2}\left(F^*(u(t))-F^*(v^\star)\right),
\end{equation}
 where $v^\star=A \lambda^\star $ is the minimizer of $F^*$ on ${\rm Im}(A)$, $\lambda^\star $ being a minimizer of $F^*_A$, $\theta,S>0$ to be defined, and $A^*$ the pseudo-inverse of matrix $A$  tuned with $\mu_{ij}^2=p_{ij}$. Let $\gamma_p$ be the smallest non-null eigenvalue of the Laplacian of the graph, with weights $\nu_{ij}=p_{ij}$ (Definition \ref{laplacian}).
\begin{thm}\label{thm:cacdm}
For the \emph{CACDM} algorithm defined by Equations \eqref{eq:contract}, \eqref{eq:x_step}, \eqref{eq:y_step} in the P.p.p. model, if $\theta=\sqrt{\frac{\gamma_p}{IS^2L_{\max}}}$ for $S$ verifying the inequality:
\begin{equation}
    S^2 \geq \sup_{(ij)\in E} \frac{(\sigma_i^{-1}+\sigma_j^{-1})}{2p_{ij}/I},
\end{equation}where $\sigma_i$ defined in \eqref{eq:smooth_sc}, $I$ in \eqref{eq:def_I} and $\sigma_{\min},L_{\max}$ in \eqref{eq:sigma_min_L_max}, we have for all $t\in\R^+$: 
\begin{equation*}
 \E[\mathcal{L}_t] \leq \mathcal{L}_0e^{-I\theta t}. 
\end{equation*}
where $\mathcal{L}_t$ is defined in \eqref{eq:lyap1}.
\end{thm}
\noindent The proof of this theorem uses edge-dual variables and differential inequalities through the study of infinitesimal intervals $[t,t+dt]$ as in the proof of Theorem \ref{thm:ppp_standard}, further combined with inequalities in \cite{neststich2017acdm} for the study of accelerated coordinate descent. We first make a few comments on this theorem, and then proceed to its proof. Since $x(t)=\nabla F^*(u(t))$ and $v^\star=\nabla F^*(x^\star)$ where $x^\star$ is the minimizer of $F$ under the consensus constraint, we have on primal variable $x(t)$:
\begin{equation}
    \E\left[\NRM{x_t-x^\star}^2\right]\le \frac{2L_{\max}}{\sigma_{\min}^2}\frac{\gamma_p}{2\theta^2S^2L_{\max}^2}\mathcal{L}_0 e^{-I\theta t}.
\end{equation}

\noindent \textbf{Remarks on the bound:} $(\gamma_p/I)$ is the normalized non-accelerated randomized gossip rate of convergence. It is divided by $I$ so that the $p_{ij}$ sum to $1$. If there exists a constant $c>0$ such that:
\begin{equation*}
    \forall (ij)\in E, \frac{p_{ij}}{I}\geq \frac{c}{|E|},
\end{equation*}then $S^2\geq \sigma_{\min}^{-1}|E|/c$, leading to the following rate of convergence:
\begin{equation*}
    I\times \sqrt{c\frac{\sigma_{\min}}{L_{\max}}\times \frac{\gamma_{{p}}}{I|E|}}.
\end{equation*}
Taking $I=1$ (re-normalizing time) and the simple averaging problem leads to an improved rate of $n^{-2}$ on the line graph instead of $n^{-3}$ \citep{Mohar1991laplacian}. For the 2D-Grid, we have $n^{-3/2}$ instead of $n^{-2}$ \citep{Mohar1991laplacian}. However, there is no improvement on the complete graph ($1/n$ in both cases). These rates are the same as \cite{Dimakis_2008, hendrikx2018accelerated,loizou2018provably}. Yet, our algorithm does not require to know the number of activations performed on the whole network, and only requires local clocks. Moreover, similarly to~\citet{hendrikx2018accelerated}, it works for any graph and for the more general problem of distributed optimization of smooth and strongly convex functions provided dual gradients of local functions are computable.
\\

\noindent \textbf{\emph{CACDM} algorithm for the averaging problem:} for the gossip averaging problem, we have $u(t)=x(t)-x(0),v(t)=y(t)-y(0)$, and \eqref{eq:x_step} and \eqref{eq:y_step} read as:
\begin{align*}
    &x_i(t)\xleftarrow t \frac{x_i(t)+x_j(t)}{2}\\
        &y_i(t)\xleftarrow t y_i(t)-\frac{\theta}{\gamma_p}\left(x_t(i)-x_t(j)\right).\label{eq:y_step}   
\end{align*}
The first variable thus performs classical local averagings while mixing continuously with the second one (the momentum).

\subsubsection{Proof of Theorem~\ref{thm:cacdm}}
Let the two edge dual variables $\lambda,\omega\in\R^{E\times d}$ be the edge-dual conjugates of $x(t),y(t)$. Variable $\omega$ plays the role of the momentum. Since $u(0)=v(0)=0$, we can take $\lambda_0=\omega_0=0$. Operations \eqref{eq:x_step} and \eqref{eq:y_step} translate as follows on these variables when clock $(ij)$ ticks. 
Let $\sigma_A$ be the strong convexity parameter of $F_A^*$ with respect to the Euclidean norm on the orthogonal of ${\rm Ker}(A)$. In Appendix \ref{app:prelim}, we prove that, if $\mu_{ij}^2=p_{ij}$: $\sigma_A\le \frac{\gamma_p}{L_{\max}}.$

While working with $F_A^*$ and edge-dual variables, we use $\sigma_A$ instead of $\frac{\gamma_p}{L_{\max}}$ as presented in the algorithm, in order to keep in mind its meaning for $F_A^*$.
Define the coordinate gradient step:
\begin{equation}\eta_{ij,t}=-\begin{pmatrix} \frac{1}{2\mu_{ij}^2(\sigma_i^{-1}+\sigma_j^{-1})}U_{ij}\nabla_{ij}F_A^*(\lambda_t) \\ \frac{\theta}{\sigma_A p_{ij}}U_{ij}\nabla_{ij}F_A^*(\lambda_t) \end{pmatrix}\label{eq:step_dual_acdm}\end{equation}
where $U_{ij}=e_{ij}e_{ij}^T$, and perform the gradient step:
\begin{equation}
    \begin{pmatrix} \lambda_{t} \\ \omega_{t} \end{pmatrix} \xleftarrow t \begin{pmatrix} \lambda_{t} \\ \omega_{t} \end{pmatrix} + \eta_{ij,t}
\end{equation}
Define:
\begin{equation*}
    L_t={\NRM{\omega_t-\lambda^\star}^*}^2 +\frac{2\theta^2S^2}{\sigma_A^2}(F_A^*(\lambda_t)-F_A^*(\lambda^\star)),
\end{equation*}
where $\NRM{.}^*$ is the Euclidean norm on the orthogonal of ${\rm Ker}(A)$, and $\lambda^\star$ is an optimizer of $F_A^*$. Note that we have $L_t=\mathcal{L}_t$ for all $t\ge0.$
\begin{proof}[Proof of Theorem \ref{thm:cacdm}] The proof closely follows the lines of~\citet{neststich2017acdm,hendrikx2018accelerated}, adapted to fit our continuous time algorithm. Without loss of generality, we assume that $I=1$ i.e. that the $p_{ij}$ sum to $1$ (by rescaling time with $t'=tI$). We denote $r_t=\NRM{\omega_t-\lambda^\star}^*$, and $f_t=F_A^*(\lambda_t)-F_A^*(\lambda^\star)$, such that $L_t=r_t^2+\frac{2\theta^2 S^2}{\sigma_A^2}f_t$. Let $t\geq0$ and $dt>0$. The following equalities and inequalities are true up to a $o(dt)$ approximation, which will disappear when we make $dt \rightarrow 0$. Let's start with the term $r_t^2$: 
\begin{align}
    \E^{\F_t}[r_{t+dt}^2]&= (1-dt)\E^{\F_t}[r_{t+dt}^2|\text{no activations between t and t+dt}]\\
    & +dt\E^{\F_t}[r_{t+dt}^2|\text{1 activation between t and t+dt}]
\end{align}
For the first term, we get:
\begin{align*}
    \E^{\F_t}[r_{t+dt}^2|\text{no activation in }[t,t+dt]] & ={\NRM{(1-\theta dt)\omega_t+\theta dt \lambda_t -\lambda^\star}^*}^2 \\
    & \leq (1-\theta dt)r_t^2 + \theta dt {\NRM{\lambda_t-\lambda^\star}^*}^2
\end{align*}
where the inequality uses convexity of the squared function. For the other term, we decompose the event "1 activation between t and t+dt" in the disjoint events "$ij$ activated between t and t+dt", of probability $p_{ij}dt$, to get the following equation, which is true up to a $o(1)$ approximation (which is enough since we multiply by $dt$ afterwards):
\begin{align}
    \E^{\F_t}&[r_{t+dt}^2|\text{1 activation between t and t+dt}]=\sum_{(ij)\in E} p_{ij} {\NRM{\omega_t-\frac{\theta}{p_{ij}\sigma_A}U_{ij} \nabla_{ij} F_A^*(\lambda_t)-\lambda^\star}^{*2}} \nonumber\\
    &=\| \omega_t - \lambda^\star\|^{*2}
    +\sum_{ij}p_{ij}\frac{\theta^2}{\sigma_A^2p_{ij}^2}{\NRM{U_{ij}\nabla_{ij}F_A^*(\lambda_t)}^*}^2
    -2\sum_{ij}p_{ij}\frac{\theta }{p_{ij}\sigma_A}\langle  U_{ij} \nabla_{ij}F_A^*(\lambda_t),\omega_t-\lambda^\star\rangle  \label{CACDM2}
\end{align}
For the term $\sum_{ij}p_{ij}\frac{\theta^2}{\sigma_A^2p_{ij}^2}{\NRM{U_{ij}\nabla_{ij}F_A^*(\lambda_t)}^*}^2$, we get by definition of $S^2$, and by a local smoothness inequality (namely, $\forall y, F_A^*(y)-F_A^*(y-\frac{1}{\mu_{ij}^2(\sigma_i^{-1}+\sigma_j^{-1})}U_{ij}\nabla_{ij}F_A^*(y))\geq \frac{1}{2\mu_{ij}^2(\sigma_i^{-1}+\sigma_j^{-1})}\NRM{\nabla_{ij}F_A^*(y)}^2$ in Lemma \ref{lemma:smoothness}):
\begin{align}
    \sum_{ij}p_{ij}\frac{\theta^2}{\sigma_A^2p_{ij}^2}{\NRM{U_{ij}\nabla_{ij}F_A^*(\lambda_t)}^*}^2 & \leq \sum_{ij}p_{ij}\frac{2\theta^2S^2}{\sigma_A^2\mu_{ij}^2(\sigma_i^{-1}+\sigma_j^{-1})}\NRM{U_{ij}\nabla_{ij}F_A^*(\lambda_t)}^2 \nonumber \\ 
    & \leq \sum_{ij}p_{ij}\frac{2\theta^2S^2}{\sigma_A^2} (F_A^*(\lambda_t)-F_A^*(\lambda_t-\frac{\theta}{\sigma_A p_{ij}}U_{ij}\nabla_{ij}F_A^*(\lambda_t))) \nonumber \\
    & =\frac{\theta^2S^2}{\sigma_A^2}(F_A^*(\lambda_t)-\E^{\F_t}[F_A^*(\lambda_{t+dt})|\text{1 activation in [t,t+dt]}]).  \label{eq:cacdm_smoothness}
\end{align}
For the term $-2\sum_{ij}p_{ij}\frac{\theta }{p_{ij}\sigma_A}\langle  U_{ij}\nabla_{ij}F_A^*(\lambda_t),\omega_t-\lambda^\star\rangle  $, we get, by adding and subtracting a $\lambda_t$ in the bracket, and by convexity of $F_A^*$ ($\sigma_A$ is the strong convexity parameter of $F_A^*$):
\begin{align*}
    -2dt\frac{\theta }{\sigma_A}\langle  \nabla F_A^*(\lambda_t),&\omega_t-\lambda^\star\rangle   = -2dt\frac{\theta }{\sigma_A}\langle   \nabla F_A^*(\lambda_t),\omega_t-\lambda_t\rangle  -2dt\frac{\theta }{\sigma_A}\langle   \nabla F_A^*(\lambda_t),\lambda_t-\lambda^\star\rangle  \\
    & \leq -2\frac{1}{\sigma_A}\langle   \nabla F_A^*(\lambda_t),\theta dt(\omega_t-\lambda_t)\rangle  -2dt\frac{\theta }{\sigma_A}(F_A^*(\lambda_t)-F_A^*(\lambda^\star)+\sigma_A/2{\NRM{\lambda_t-\lambda^\star}^*}^2)
\end{align*}
Then, let's define $\lambda'_{t+dt}=(1-\theta dt)\lambda_t+\theta dt \omega_t=\E^{\F_t}[\lambda_{t+dt}|\text{no activations in }[t,t+dt]]$. By noticing that $\theta dt(\omega_t-\lambda_t)=\lambda'_{t+dt}-\lambda_t$, we get:
\begin{align}
    -2\frac{1}{\sigma_A}\langle   \nabla F_A^*(\lambda_t),\theta dt(\omega_t-\lambda_t)\rangle   & = -2\frac{1}{\sigma_A}\langle   \nabla F_A^*(\lambda_t),\lambda'_{t+dt}-\lambda_t\rangle   \label{cacdm1}\\
    & = -2\frac{1}{\sigma_A}\langle   \nabla F_A^*(\lambda'_{t+dt}),\lambda'_{t+dt}-\lambda_t\rangle  \label{cacdm2}\\
    & \leq -2\frac{1}{\sigma_A}(F_A^*(\lambda'_{t+dt})-F_A^*(\lambda_t)),\label{cacdm3}
\end{align}
where from \eqref{cacdm1} to \eqref{cacdm2}, the equality holds at $o(dt)$, as the left part of the bracket is true at $o(1)$ precision, and the right part of the bracket is a $O(dt)$. Then, we go from \eqref{cacdm2} to \eqref{cacdm3} using the convexity of $F_A^*$.
By plugging~\eqref{eq:cacdm_smoothness} and~\eqref{cacdm3} into~\eqref{CACDM2}, and rearranging the terms, we obtain:
\begin{align*}
    \E^{\F_t}[r_{t+dt}^2]-r_t^2 & \leq -dt \theta r_t^2 +dt\theta {\NRM{\lambda_t-\lambda^\star}^*}^2\\
    & +dt \frac{2\theta^2S^2}{\sigma_A^2}(F_A^*(\lambda_t)-\E^{\F_t}[F_A^*(\lambda_{t+dt})|\text{1 activation in [t,t+dt]}]) \\
    & -2dt\frac{\theta }{\sigma_A}(F_A^*(\lambda_t)-F_A^*(\lambda^\star)+\sigma_A/2{\NRM{\lambda_t-\lambda^\star}^*}^2)-2\frac{1}{\sigma_A}(F_A^*(\lambda'_{t+dt})-F_A^*(\lambda_t))
\end{align*}
Studying $\E^{\F_t}[F_A^*(\lambda_{t+dt})]$, we get:
\begin{align}
    \E^{\F_t}[F_A^*(\lambda_{t+dt})] & =(1-dt)F_A^*(\lambda'_{t+dt})+ dt \E^{\F_t}[F_A^*(\lambda_{t+dt})-F_A^*(\lambda^\star)|\text{1 activation in [t,t+dt]}]
\end{align}
Using $\theta^2=\sigma_A/S^2$ (i.e $\theta^2S^2/\sigma_A^2=1/\sigma_A$) and the above equality, Equation~\eqref{CACDM2} become:
\begin{align*}
    \E^{\F_t}[r_{t+dt}^2]-r_t^2 & \leq -dt \theta r_t^2 \\
    & +dt \frac{2}{\sigma_A}(F_A^*(\lambda_t)-\E^{\F_t}[F_A^*(\lambda_{t+dt})|\text{1 activation in [t,t+dt]}]) \\
    & -2dt\frac{\theta }{\sigma_A}(F_A^*(\lambda_t)-F_A^*(\lambda^\star ))-2\frac{1}{\sigma_A}(F_A^*(\lambda'_{t+dt})-F_A^*(\lambda_t))\\
    & = -dt \theta r_t^2  -\frac{2}{\sigma_A}(\E^{\F_t}[F_A^*(\lambda_{t+dt})-F_A^*(\lambda^\star )]-F_A^*(\lambda_{t})-F_A^*(\lambda^\star ))\\
    & -2dt\frac{\theta }{\sigma_A}(F_A^*(\lambda_t)-F_A^*(\lambda^\star )) +2dt\frac{\theta }{\sigma_A}(F_A^*(\lambda_t)-F_A^*(\lambda'_{t+dt}))\label{cacdm44}
\end{align*}
Since the last term is a $o(dt)$, the previous equation simplifies to:
\begin{equation*}
    \E^{\F_t}[L_{t+dt}]-L_t\leq - \theta dt L_t
\end{equation*}
Finally, we take the expectation, divide by $dt$ and make it tend to zero, which leads to $\frac{d}{dt}\E L_t\leq -\theta \E L_t$. Integrating this leads to the desired result, which writes: $$\forall t \geq 0, \E L_t\leq \exp{(-\theta t)}L_0$$
\end{proof}

\noindent \textbf{Empirical Results:} 
We consider the \emph{P.p.p.~model} on two graphs: the circle with 50 nodes and the 2D-Grid with 225 nodes. Our goal is to illustrate how the algorithms compare in a heterogeneous setting. Therefore, in both cases, $10\%$ of the nodes (chosen uniformly at random) have a delay $\tau_i=100$ time units, while the others have a delay equal to $1$ time unit. The delay of an edge $(ij)$ is then $\tau_{ij}=\max(\tau_i,\tau_j).$ Then, we take \emph{Poisson} rates for the edges equal to the inverse of these delays: $p_{ij}=1/\tau_{ij}$. The local functions for the gossip problems are chosen as $f_i(x)=\NRM{x-c_i}^2$, with $c_0 = 1$ and $c_i=0$ otherwise, which is the worst case scenario in terms of mixing). Figure~\ref{fig:ppp} shows the performances of classical (pairwise) gossip and \emph{CACDM} in this setting. We see that \emph{CACDM} is much faster than classical gossip, and that this is true in particular when the eigengap of the graph is small (of order $1/125000$ for our cyclic graph, compared to $1/50000$ for our grid), as predicted by Theorem \ref{thm:cacdm}.\\

\begin{figure*}[h!]
\subfigure[2D-Grid with 225 nodes]{
   \includegraphics[width=0.43\linewidth]{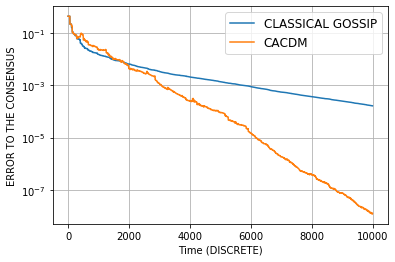}
}
\hfill
\subfigure[Cyclic graph with 50 nodes]{
    \includegraphics[width=0.43\linewidth]{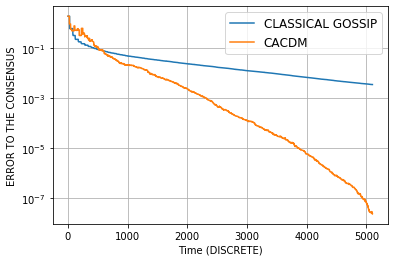}
}
\vspace{-8pt}
\caption{CACDM vs Randomized Gossip in the \emph{P.p.p.~model}.
}
\label{fig:ppp}
\end{figure*}

\section{Gossip on Loss Networks \label{section:LN}}

\subsection{Refined Loss-Network Communication Scheme and Detailed Algorithm \label{LNpresentation}}

The \emph{P.p.p.~model} is particularly amenable to analysis, and 
helps us understand quantitatively why asynchronous algorithms can outperform synchronous ones, but it assumes that communications and computations are done instantaneously. Thus, actual implementations differ from its underlying assumptions, unless further synchrony is assumed~\citep{hendrikx2018accelerated}. To alleviate this issue, with pairwise communications ruled by point processes as a baseline, we consider a protocol in which nodes are tagged as \emph{busy} when they are already engaged in an update, and communications between busy nodes are forbidden. Our model is inspired from classical Loss Network models \citep{kelly1991lossnetwork}.
In our new model, edges are activated following the same procedure as in the \emph{P.p.p.~model}, with a \emph{P.p.p.}~of intensity $p_{ij}$. Note that we do not consider these intensities to be constraints of the problem, but rather parameters of the algorithm, that we tune below. Each node has an exponential clock of intensity $\frac{1}{2}\sum_{j\sim i}p_{ij}$. At each clock-ticking, if $i$ is not busy, it selects a neighbor $j$ with probability $p_{ij}/\sum_{k\sim i}p_{ik}$. Node $i$ first checks if $j$ is currently \emph{busy}, an operation that takes time $\eps\tau_{ij}$ for some \emph{small} $\eps>0$ ($\eps\ll1$ if sending a simple request if much faster than sending a whole vector). If $j$ is not \emph{busy}, $i$ and $j$ compute and exchange information, becoming busy for a duration $\tau_{ij}$. We can think of this procedure as classical gossip on an underlying random graph (Figure \ref{fig:lossnetwork_dessin}), that follows a Markov-Chain process if we extend the space of states with the inactivation time. We call our model the \textbf{\emph{Refined Loss Network Model of parameter $\eps$} (\emph{RLNM($\eps$}))}. It is \emph{refined} as the operation that consists in checking on its neighbors is not present in classical Loss Networks.\\
\begin{figure}[h]
    \centering
    
    \includegraphics[scale=0.7]{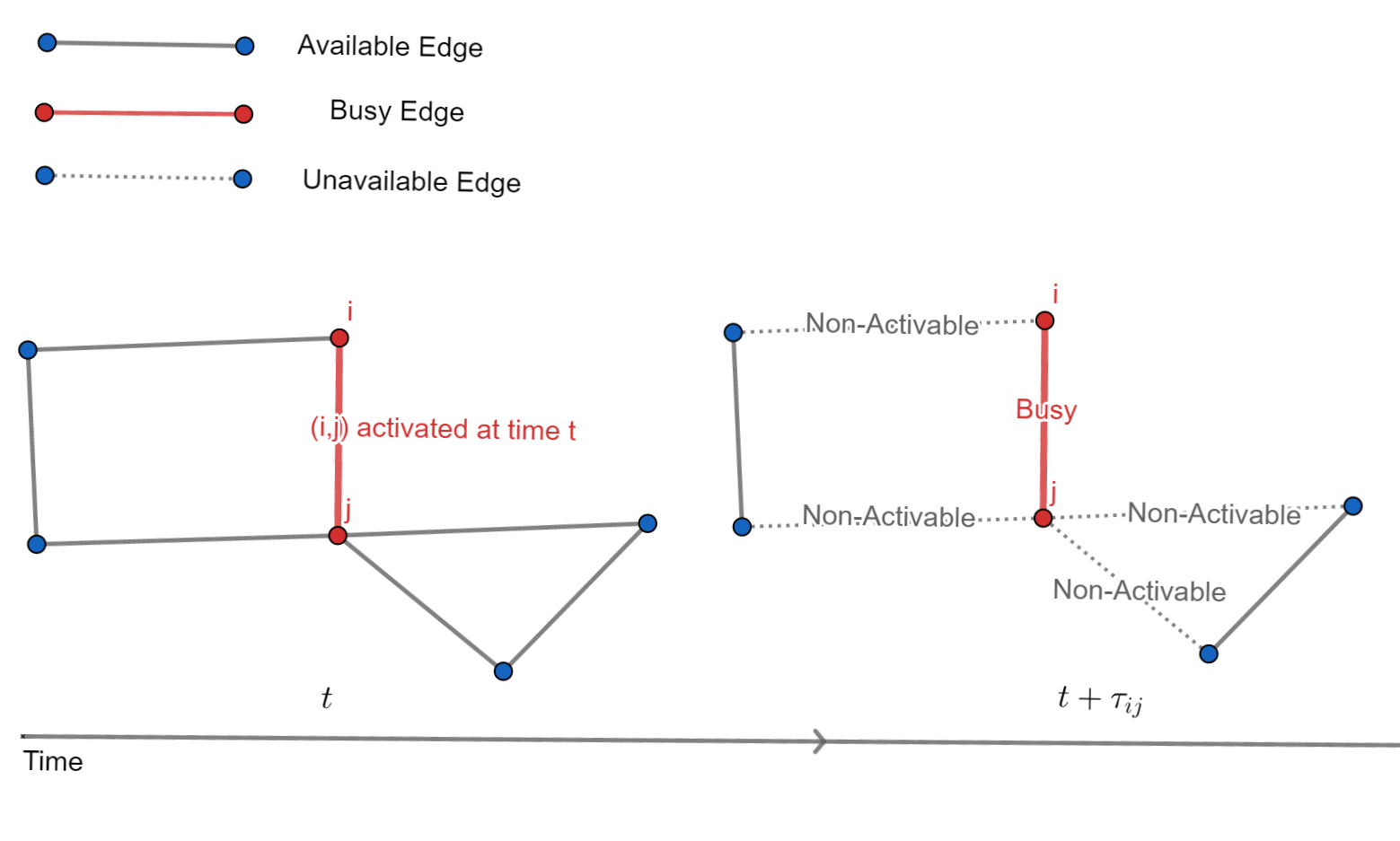}

    \caption{Underlying Markov Process for the Graph: edge $(ij)$ activated at time $t$ implies that while $ij$ busy \emph{i.e.}~between times $t$ and $t+\tau_{ij}$, all edges $kl$ adjacent to $ij$ are unavailable.}
    \label{fig:lossnetwork_dessin}
\end{figure}

More precisely, asynchronous gossip on the Refined Loss-Network communication model runs as follows:
each node has a local clock and a \emph{Poisson Point Process} of intensity $\frac{1}{2}\sum_{j\sim i} p_{ij}$, where, with $d_i$ the degree of node $i$ and $\tau_{\max}(ij)=\max_{kl\sim ij}\tau_{kl}$:
\begin{equation}
    p_{ij}=\min\left(\frac{1}{\tau_{\max}(ij)},\frac{1}{2(\max(d_i,d_j)-1)}\frac{1}{\tau_{ij}}\right).\label{eq:pppLN}
\end{equation}
Let $I=\sum_{ij\in E} p_{ij}$ the global activation intensity. Let node $i$'s local variable be $x_i(t)$ at time $t$, and let $v_i(t)=\nabla f_i(x_i(t))$ its convex-conjugate. Note $\lambda_t$ the edge-dual variable at time $t$ (we have $A\lambda_t=v(t)$). Initialize such that $v(0)=0$ (and $\lambda_0=0$).
\begin{enumerate}
    \item \textbf{ "Busy-Checking" Operation:} when clock $i$ rings at time $t$, select $j\sim i$ with probability $\frac{p_{ij}}{\sum_{k\sim i}p_{ki}}$ and check whether $j$ is busy. This operation makes $i$ busy for a timelapse of length $\eps \tau_{ij}$.
    \item \textbf{Gradient Exchange:} if neighbor $j$ (chosen at the previous step) is not busy, make both nodes busy for a time $\tau_{ij}$, and $i$ sends $\nabla f_i^*(v_i(t))$ to $j$ (and reciprocally).
    \item \textbf{Gradient Step:} when $i$ receives gradient $\nabla f_j^*(v_j)$ from $j$, it updates its local value $v_i$ using the following gradient step:
\begin{equation}\label{eq:LN_update}
    v_i(t) \xleftarrow t v_i(t) - \frac{\nabla f_i^*(v_i(t))-\nabla f_j^*(v_j(t))}{\sigma_i^{-1}+\sigma_j^{-1}}.
\end{equation}
\end{enumerate}
The desired output at node $i$ at time $t$ is then $x_i(t)=\nabla f_i^*(v_i(t))$. Note that in the gossip averaging problem, these operations are equivalent to local averagings as shown in \eqref{eq:gossip_update}.
Operations (2) and (3) both happen in the timelapse of length $\tau_{ij}$, thus causing no asynchrony issues and avoiding the need to consider delayed gradients.

\subsection{Convergence Results}

Define the following constants, where $p_{ij}$ is set as in~\eqref{eq:pppLN} and $I$ defined in \eqref{eq:def_I}:
\begin{equation}
\left\{\begin{array}{lll}
    \Tilde{\tau}_{ij}&=(1+\eps)p_{ij}^{-1}\\
     \Tilde{\tau}_{\max}&=\max_{(ij)\in E} \Tilde{\tau}_{ij}\\
     T&= \frac{2\log(6|E|)}{\log\left(1-(1-e^{-1})e^{-1}\right)}I\Tilde{\tau}_{\max}
\end{array}.\right.\label{eq:quant_LN}
\end{equation}

\noindent Define for $k\in\N$, $\mathcal{E}_k=F^*(v(t_k))-F^*(v^\star)$, where $v^\star=A \lambda^\star $ is a minimizer of $F^*$ on ${\rm Im}(A)$, $\lambda^\star $ being a minimizer of $F^*_A$ and $t_k\in\R^+$ is the time of the $k$-th activation. For $k\in \N$, let $\mathcal{L}_k$ be the following Lyapunov function:
\begin{equation}\label{eq:lyapu_LN_v}
    \mathcal{L}_k=\frac{1}{T}\sum_{l=k}^{k+T-1}\mathcal{E}_l.
\end{equation}
This choice of Lyapunov function is motivated by the fact that we want to take into account $T$ successive values of $\mathcal{E}_l$ (the dual error to the optimum), where $T$ is the typical number of activations required to have all edges activated.
Note that this Lyapunov function bears some resemblance with Lyapunov-Krasovskii functionals (see e.g. \cite{fridman2001timedelay}) used in the study of delayed differential systems, and which can be thought of as the continuous analog of $\mathcal{L}_k$, with an integral instead of a sum. We insist on the fact that considering this specific Lyapunov is a key step of our proof. 

\begin{thm}[Discrete-time rate of convergence in the Loss-Network model]\label{thm:LN}
Consider the \emph{CDM} algorithm \eqref{eq:LN_update}, with node activations according to the \emph{RLNM} with Poisson rates \eqref{eq:pppLN}. 
Let $\Gamma_{RLNM}=\gamma(\nu_{ij})$ be the spectral gap of the weigthed graph Laplacian   with weights 
\begin{equation*}
  \nu_{ij}= \alpha \times \frac{\Tilde{\tau}_{ij}^{-1}\min_{(kl)\sim (ij)}\frac{\Tilde{\tau}_{ij}}{\Tilde{\tau}_{kl}}}{Id_{\max}^2\left(\log(|E|)+\log(I\Tilde{\tau}_{\max})\right)^2} \label{LNrate},
\end{equation*}where $\alpha=\frac{32e^2}{\log(1-(1-e^{-1})e^{-1})^2}$ is a universal constant and $d_{\max}$ is the maximal degree in the graph. Then, for all $k\in \N$:
\begin{equation*}
\label{eq:lyapunov_stochasticity_LN}
    \E[\mathcal{L}_k]\leq \left(\frac{1}{4}(1-\frac{\sigma_{\min}}{L_{\max}}\Gamma_{RLNM})^{T/3}+\frac{3}{4}\right)^{\lceil \frac{k}{2T}\rceil}\E[\mathcal{L}_0].
\end{equation*}
where Lyapunov function $\mathcal{L}_k$ is defined in \eqref{eq:lyapu_LN_v}. 
\end{thm}
Theorem~\ref{thm:LN} gives precise results in a general setting, but it may be hard to parse. In order to present results in a more concise form, we introduce the simplifying Assumption~\ref{hyp:delay_constraint}, which in particular allows to obtain an asymptotic rate of convergence for $\mathcal{E}_k$.
\begin{hyp}[Delay Constraints]\label{hyp:delay_constraint}
Let $\gamma_1=\gamma(\nu_{ij})$ for $\nu_{ij}\equiv 1,(ij)\in E$ (Definition \ref{laplacian}). Assume that:
\begin{equation}
    \frac{\Tilde{\tau}_{\max}}{\Tilde{\tau}_{\min}} \leq \frac{L_{\max}}{\sigma_{\min}}\times \frac{\alpha d_{\max}^2\log(|E|)}{\gamma_1}.\label{eq:delay_constraint}
\end{equation}
\end{hyp}
Notice that the right-hand side of \eqref{eq:delay_constraint} reflects the complexity of the optimization problem through the first factor (generally referred to as the 
condition number of the optimization problem), and the topology of the graph (without the delays) through $\gamma_1$. The more difficult the problem is, the bigger the right-hand side is. Assumption \ref{hyp:delay_constraint} will then be verified more easily for graphs with slow mixing times ($\gamma_1^{-1}$ bigger) and less regular local functions. The order of magnitude of $\gamma_1^{-1}$ is $n^2$ for the grid, and $n$ for the line or the cyclic graph. More generally, the right-hand side of \eqref{eq:delay_constraint} is always of order bigger than $n$.

\begin{cor}[Asymptotic Rate] \label{cor:LN}Under Assumption \ref{hyp:delay_constraint}, Theorem \ref{thm:LN} gives: 
\begin{equation*}
    \limsup_{k\to \infty} \frac{1}{k}\log\left(\E[\mathcal{E}_k]\right) \leq -\frac{\sigma_{\min}}{L_{\max}} \times \frac{\Gamma_{RLNM}}{24e}.
\end{equation*}
\end{cor}
\noindent \textbf{Comments on the convergence rate:} Theorem \ref{thm:LN} and Corollary \ref{cor:LN} are formulated in discrete time. The continuous exponential rate of convergence is obtained by multiplying by the global \emph{P.p.p.}~intensity $I$, up to a constant factor of order $1$. The factor $\frac{1}{I}$ in the definition \eqref{LNrate} of the weights $\nu_{ij}$ is hence simply a normalization factor, due to a study in discrete time.
As desired, the communication cost factor in the rate of convergence ($\Gamma_{RLNM}$) is captured by the Laplacian of the graph, weighted by \emph{local} delays, instead of $\tau_{\max}^{-1}$. We however observe slowdowns due to other factors. 
\begin{enumerate}
    \item Having $\Tilde{\tau}_{ij}$ instead of $\tau_{ij}$ (as in the $\emph{P.p.p.~model}$ \eqref{eq:asynch_gossip}) means that the effective waiting time of edge $ij$ between two activations is of order $\Tilde{\tau}_{ij}$ (defined in \eqref{eq:quant_LN}) and not $\tau_{ij}$, which was expected since $p_{ij}$ is tuned accordingly. 
    \item Adding the factor $\min_{(kl)\sim (ij)}\frac{\Tilde{\tau}_{ij}}{\Tilde{\tau}_{kl}}$ to the local weight in the Laplacian is a local slowdown: a node with a slow neighbor becomes less effective. 
    
    These first two remarks 1) and 2) suggest that by deleting some edges one could improve the rate of convergence. A similar phenomenon occurs in road-trafficking \citep{kelly1997braessparadox,steinberg1983braess}, where deleting some roads can lead to reduced congestion (\emph{Braess's paradox}).
    \item The global factor $\frac{1}{d_{\max}}$ is not intuitive at first: the more connected the graph is, the higher the rate should be. We hence have a trade-off between $\frac{1}{d_{\max}}$ that decreases when adding edges, and the smallest eigenvalue of the Laplacian of the graph $\Gamma$ that increases with connectivity. We believe that $\frac{1}{d_{\max}}$ is an artifact of the proof, but have not been able so far to remove it.
    \item If some nodes are \emph{stragglers} (\emph{i.e.}~with high delays compared to the others), the rate of convergence stated for \emph{RLNM} improves over synchronous algorithms, as it takes into account \emph{local} delays. If the delays are all of the same order of magnitude, a case favorable to synchrony, the rate obtained is the same as in synchronous algorithms, up to a factor of order $\frac{1}{d^2\log(n)}$. The factor $d^2$ should not be of too much importance in $d$-regular graphs for $d\ll n$, such as grids or lines. The $\log$ factor comes from exponential tails of our random variables. 
\end{enumerate}

\begin{remark}[Comparison with a delayed information approach] One may wonder how our model compares to a delayed information approach, in which nodes send gradients whenever they can. In the delayed information approach, delays increase the variance of the gradients, typically by a multiplicative factor $\tau$ equal to the discrete-time delay  \citep{leblond2016asaga,hannah2018a2bcd} thus requiring step sizes to be scaled by the inverse of the delays. However, the few works done in this direction rely on a global upper-bound $\tau_{\max}$ on the delays, and as such  provide slow rates in scenarios with heterogeneous local delays, compared to those achievable with our \emph{RLNM} approach. Developing delayed information schemes that are competitive in heterogeneous scenarios is an open research direction.
\end{remark}

\subsection{Sketch of Proof of Theorem~\ref{thm:LN}}
This proof follows three main steps: \emph{i)} Deriving convergence results for more general communication schemes than \emph{RLNM}, under deterministic assumptions on the delays. \emph{ii)} Adapting Step i) to stochastic assumptions on the delays. \emph{iii)} Deriving high-probability upper-bounds on the delays between two activations in \emph{RLNM} in order to fall under the assumptions of Step i).

As in the previous proofs, the analysis is done with edge-dual variable $\lambda_t\in\R^{E\times d}$, such that $A\lambda_t=v(t)$. Matrix $A$ is tuned in the detailed proof (Appendix \ref{app:LN}). When nodes $i\sim j$ exchange gradients, it is equivalent to, on edge $(ij)$:
\begin{equation}\label{eq:tmptmp}
    \lambda_t\xleftarrow t \lambda_t -\frac{1}{\mu_{ij}^2 (\sigma_i^{-1}+\sigma_j^{-1})}\nabla_{ij}F_A^*(\lambda_t).
\end{equation}
\subsubsection{Step 1: General Communication Schemes\label{section:step1}} We consider general activation processes $\cP_{ij}$. When edge $(i,j)$ is activated, the update described in \eqref{eq:LN_update} is performed at nodes $i$ and $j$. The delay of an edge is defined as its (random) waiting time between two activations.
Two ergodicity-like conditions on the delays are needed: (i) edges activated regularly enough and (ii) incident edges must not be activated too many times. We now formally introduce these assumptions. We consider discrete time in this section: more precisely, $t\in \N$ stands for the $t$-th edge activation.
\begin{defn}\label{def:quant_delays} Consider a communication scheme with edge-activation point processes $\cP_{ij}$. Let $t=0,1,2,...$ index the consecutive edge activations. Let $s\in \N$, $ij$ and $kl \in E$. Let $s_{ij}<t_{ij}$ such that $s_{ij}\leq s <t_{ij}$ be consecutive activation times (in discrete time) of $(ij)$. Denote $T_{ij}(s)=t_{ij}-s_{ij}-1$ the total number of edge activations between the two consecutive activations of $ij$. Denote $N(kl,ij,s)$ the number of activations of edge $kl$ in the activations $\{s_{ij},s_{ij}+1,...,t_{ij}-1\}$.
\end{defn}

\begin{hyp}[Delay Assumptions] \label{hyp} There exist $T\in \N^*$, $a,b>0$, and $\ell_{ij}>0,ij\in E$ such that, for the quantities and the communication scheme in Definition \ref{def:quant_delays}:
\begin{enumerate}
    \item For all $t \in \N$, all edges are activated between iterations $t$ and $t+T-1$.
    \item $\forall s\geq0, \forall (ij)\in E, T_{ij}(s)\leq a \ell_{ij}$: $(ij)$ is activated at least every $a\ell_{ij}$ activations.
    \item $\forall s\geq 0, \forall (ij),(kl)\in E$ such that $(kl)\sim(ij)$, $N(kl,ij,s)\leq \lceil \frac{b\ell_{ij}}{\ell_{kl}}\rceil$.
\end{enumerate}
\end{hyp}
Assumption (1) is implied by Assumption (2) if $T=\max_{(ij)}\ell_{ij}$. Taking $\ell_{ij}$ as a deterministic upper-bound on the delays of edge $(ij)$ between two activations in continuous time is sufficient to have Assumption (2) and (3), with some normalizing constant $a$, and $b$ such that $\ell_{ij}/b$ is a lower-bound on these delays.

The main technical difficulty lies in the fact that at a defined activation time $t$, some nodes are not available: at any time $t\ge 0$, $\sum_{(ij)\in E \text{ not busy}} \nabla_{ij} F_A^*(\lambda_t)$ usually differs from $\nabla F_A^*(\lambda_t)$ as in \emph{Markov-Chain Gradient Descent} \citep{sun2018mcgd}, thus making an analysis such as in the \emph{P.p.p. model} impossible.
To alleviate this difficulty, in order to make sure that all edges are taken into account when performing the averaging, the Lyapunov function $\Lambda_t$ that we study considers the value of the objective for $T$ consecutive activation times. It is defined as follows on the dual variable:
$\forall t\in \N, \Lambda_t=\frac{1}{T}\sum_{s=t}^{t+T-1} F_A^*(\lambda_s)-F_A^*(\lambda^\star ).$ Note that we have $\Lambda_t=\mathcal{L}_t$ for any $t\in\N$, $\mathcal{L}_t$ as in \eqref{eq:lyapu_LN_v}: we simply changed notations as we work with edge-dual variables, and time is indexed in a different way. The first step of the proof of Theorem \ref{thm:LN} consists in proving the following. A detailed proof of this can be found in Appendix \ref{app:general_activations}.
\begin{thm} 
\label{thmdet}
Consider a general communication scheme as in Definition \ref{def:quant_delays}, that satisfies Assumption \ref{hyp} for constants $\ell_{ij},a,b>0,$. At every edge-activation of edge $(ij)$, update \eqref{eq:tmptmp} is performed. Let $\gamma$ be the smallest positive eigenvalue of the Laplacian of the graph with:
\begin{equation*}\nu_{ij}=C\ell_{ij}^{-1}\min_{kl\sim ij}\frac{\ell_{kl}}{\ell_{ij}},
\end{equation*}
where $C=\frac{1}{2a+8d_{\max}^2ab}$. Then, we have, for $t\in \N$:
\begin{equation*}
    \Lambda_t\leq \left(1-\frac{\sigma_{\min}}{L_{\max}}\times\gamma\right)^t \Lambda_0.
\end{equation*}
\end{thm}

\subsubsection{Step 2: Introducing Stochasticity\label{section:step2}} Theorem \ref{thmdet} cannot be applied directly to \emph{RLNM} since we have unbounded delays. Yet, Theorem~\ref{thmdet} can be adapted to hold with relaxed assumptions: the conditions on the delays may only hold with some (not too low) probability instead of almost surely. More precisely, we prove the following in Appendix \ref{app:LN_sto}.
\begin{prop}[Adding Stochasticity ]\label{thmsto} Assume that, for all $t\in \N$, there exists a $\F_{t+T-1}$-measurable event $A_t$, such that $\P(A_t|\F_t)\geq \frac{1}{2}$ almost surely, and that under $A_t$, Assumption \ref{hyp} holds  for $t\leq s \leq t+T-1$. Then, we have the following bound on $L_t$, : 
\begin{equation*}
\label{eq:lyapunov_stochasticity}
    \E[\Lambda_t]\leq \left(\frac{1}{4}(1-\frac{\sigma_{\min}}{L_{\max}}\gamma)^{T/3}+\frac{3}{4}\right)^{\lceil \frac{t}{2T}\rceil}\E[\Lambda_0].
\end{equation*}
\end{prop}
This proposition enables us to apply Theorem \ref{thmdet} to stochastic communication schemes that have unbounded yet stochastically controlled delays. This result and its proof are thus of independent interest: it encompasses more general communication schemes than \emph{RLNM}. Furthermore, the methodology of this deterministic to stochastic conversion could be applied more generally to other problems.

\subsubsection{Step 3: Controlling Inactivation Times in \emph{RLNM($\eps$)}\label{section:step3}} After studying general deterministic (Section \ref{section:step1}) and then stochastic communication schemes (Section \ref{section:step2}), we place ourselves back in the \emph{RLNM($\eps$)} model. The following lemma controls how long a given edge can remain inactive in our model, which is a key step of our analysis. Indeed, it allows us to specify the constants $\ell_{ij},T,a$, and $b$ from Assumption \ref{hyp} such that Proposition \ref{thmsto} can be applied.
\begin{lem}\label{lem_queue_1}
For any $t_0\geq 0$, $ij\in E$, if the \emph{Poisson} intensities are such that $p_{ij}=\frac{1}{2\max(d_i,d_j)-1}((1+\eps)\tau_{ij})^{-1}$ and $\tau_{max}(ij)=\max_{kl\sim ij}\tau_{kl}$, let:
\begin{equation*}
    \ell_{ij}=\frac{\log(\delta^{-1})}{\log(1-(1-e^{-1})e^{-1})}(p_{ij}^{-1}+\tau_{max}(ij))(1+\eps)
\end{equation*} for any $\delta\in (0,1)$. We have:
\begin{equation}
    \P(ij\text{ not activated in $[t_0,t_0+\ell_{ij}]$}|\F_{t_0})\leq \delta.
\end{equation}
\end{lem}
\begin{proof}[Proof of Lemma \ref{lem_queue_1}]Let $ij\in E$ and $t_0\geq 0$ fixed. We use tools from queuing theory \citep{Tanner1995queuing} ($M/M/\infty/\infty$ queues) in order to compute the probability that edge $ij$ is activable at a time $t$ or not. More formally, we define a process $N_{ij}(t)$ with values in $\N$, such that $N_{ij}(t_0)=1$ if $ij$ non-available at time $t_0$ and $0$ otherwise. Then, when an edge $kl$ such that $kl\sim ij$ is activated, we make an increment of $1$ on $N_{ij}(t)$ (a \emph{customer} arrives). This customer stays for a time $\tau_{kl}(1+\eps)$ and when he leaves, $N_{ij}$ is decreased by $1$. Thus $N_{ij}\geq 0$ a.s., and if $N_{ij}=0$, then edge $ij$ is available. For $t\geq \max_{kl\sim ij} \tau_{kl}(1+\eps) +t_0$, $N_{ij}(t)$ follows a Poisson law of parameter $\sum_{kl\sim ij}p_{kl}\tau_{kl}(1+\eps)$.  For any $t\geq \max_{kl\sim ij} \tau_{kl}(1+\eps) +t_0$:
\begin{equation*}
    \P(ij\text{ available at time }t|\F_{t_0})\geq \P(N_i(t)= 0)=\exp(-\sum_{kl\sim ij}p_{kl}\tau_{kl}(1+\eps)).
\end{equation*}
That leads to taking $p_{kl}=\frac{1}{2}\frac{1}{\max(d_k,d_l)-1}((1+\eps)\tau_{kl})^{-1}$ for all edges, in order to have $$\P(ij\text{ available at time }t|\F_{t_0})\geq 1/e.$$ Then, $\P(ij \text{ rings in } [t,t+p_{ij}^{-1}])=1-e^{-1}$, giving:
\begin{align*}
    \P&(ij \text{ activated in }[t_0,t_0+(1+\eps)\tau_{\max}(ij)+p_{ij}^{-1}]|\F_{t_0})= \P(ij \text{ rings in } [t,t+p_{ij}^{-1}])\\
    &\times \P(ij \text{ available at time }t|\F_{t_0},\text{$ij$ rings at a time } t\in [t_0+(1+\eps)\tau_{\max}(ij),t_0+(1+\eps)\tau_{\max}(ij)+p_{ij}^{-1}])\\
    & \geq (1-e^{-1})e^{-1},
\end{align*}
where we use the memoriless property of exponential random variables. Take $k\in \N$ such that $(1-(1-e^{-1})e^{-1})^k\leq \delta$, leading to $k= \log(6|E|)/\log(1-(1-e^{-1})e^{-1})$. Let $$\ell_{ij}=k(p_{ij}^{-1}+\tau_{max}(ij)(1+\eps)).$$ Then we have a.s.:
\begin{equation}
    \P(ij\text{ not activated in $[t_0,t_0+\ell_{ij}]$}|\F_{t_0})\leq \delta.
\end{equation}
\end{proof}
\noindent We then use this lemma in Appendix \ref{app:LN_tuning} in order to tune the constants of Assumption \ref{hyp} for \emph{RLNM}.

\subsection{Empirical Results}
\noindent  The results in Figure \ref{fig:lossnetwork} correspond to the Loss-Network scheme on the same two heterogeneous graphs (50-node cycle and 225-node 2D-grid) as in Figure~\ref{fig:ppp}. We compare our algorithm on the Loss-Network to synchronous gossip. Time is indexed in a continuous way. Synchronous iterations are done every $100$ units of time. The speed-up is significant when the fluctuation in term of delays in the graph is high, which illustrates the discussion at the end of Section \ref{section:formul_def}.
\begin{figure*}[h!]
\subfigure[2D-Grid with 225 nodes]{
   \includegraphics[width=0.43\linewidth]{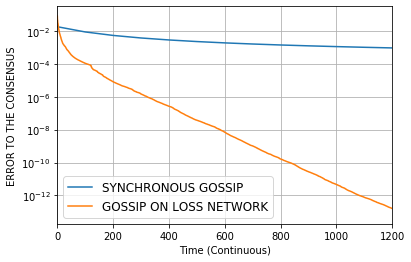}
}
\hfill
\subfigure[Cyclic graph with 50 nodes]{
    \includegraphics[width=0.43\linewidth]{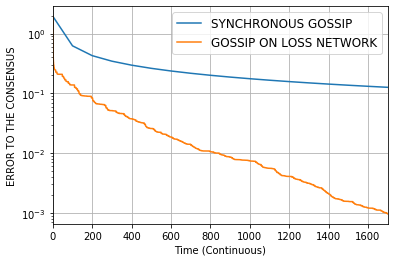}
}
\vspace{-8pt}
\caption{Asynchronous Speed-Up: Classical synchronous gossip (Appendix \ref{app:synch_gossip}) VS Gossip on \emph{RLNM}.
}
\label{fig:lossnetwork}
\end{figure*}  

\noindent \textbf{Acceleration in \emph{RLNM($\eps$)}:} The analysis of $\emph{CACDM}$ does not extend to more general models than the \emph{P.p.p.~model}. However, applying it to \emph{RLNM} leads to an accelerated rate of convergence displayed in Figure \ref{fig_CACDM_LN}, showing that our algorithm is quite robust to changes in edge activation statistics. In order to tune the algorithm, we take values $p_{ij}$ as in \eqref{eq:pppLN}. Time is indexed in a continuous way. 1000 units of time hence correspond to approximately $I\times 1000\approx 10^5-10^6$ edge activations. 
\begin{figure*}[h!]
\subfigure[2D-Grid with 225 nodes]{
   \includegraphics[width=0.43\linewidth]{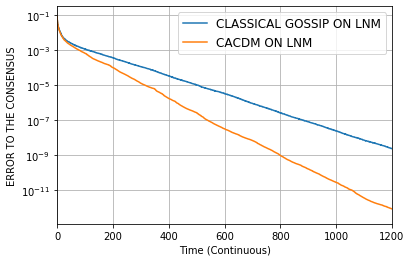}
}
\hfill
\subfigure[Cyclic graph with 50 nodes]{
    \includegraphics[width=0.43\linewidth]{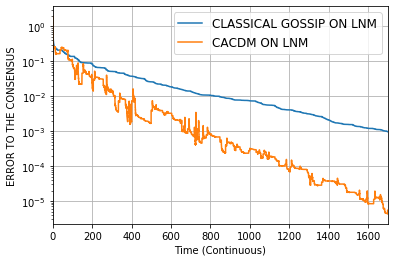}
}
\vspace{-8pt}
\caption{\emph{CACDM} vs Gossip in  \emph{RLNM}.
}
\label{fig_CACDM_LN}
\end{figure*}  

\section{Conclusion}

We studied asynchronous gossip algorithms in two frameworks: the popular \emph{P.p.p.~model} and  the \emph{refined loss network} model, a contribution of this paper. For the simple \emph{P.p.p.~model} of asynchronous operations we developed a novel analysis in continuous time of gradient descent which then enabled us to propose \emph{CACDM}, a provably accelerated version of classical randomized gossip. \emph{RLNM}, our refined model of asynchronous communications, provides a more realistic model of asynchrony than the P.p.p. model, as well as a framework that avoids the need to rely on delayed information. 
We obtained convergence rate guarantees for the \emph{CDM} scheme under this model, that highlight the role of quantities such as local effective delays, local differences of delays, and node degrees. An interesting open question is whether our established rates of convergence enjoy some form of optimality, or how fundamental the local effective delays we identified, and the spectral gap of the associated weighted graph Laplacian, are intrinsic bottlenecks for the performance of asynchronous distributed optimization. 
We believe that both our main contributions (\emph{CACDM} and \emph{RLNM}) pave the way for fast asynchronous gossip algorithms with theoretical guarantees. 

\nocite{*}
\bibliographystyle{apalike} 
\bibliography{biblio}
\renewcommand*\contentsname{Summary of the Article and of the Appendix}

\tableofcontents
\appendix
\section{Gossip Algorithms: General Considerations on the Averaging Problem\label{app:gossip}}

\subsection{Synchronous Gossip\label{app:synch_gossip}}
In the synchronous setting, all nodes are allowed to share a common clock, which enables them to perform operations synchronously. Formally, a \emph{gossip matrix} is defined as follows:
\begin{defn}[Gossip Matrix] A gossip matrix is a matrix $W\in \R^{n\times n}$ such that:
\begin{itemize}
    \item $\forall (i,j)\in [n]^2$, $W_{i,j}>0\implies i\sim j$ or $i=j$ (supported by $G$),
    \item $\forall i \in [n], \sum_{j\sim i} W_{i,j}= 1$ (stochastic),
    \item $\forall (i,j)\in [n]^2, W_{i,j}=W_{j,i}$ (symmetric).
\end{itemize}
\end{defn}
Iteratively, at times $t=0,1,2,...$, if $x(t)=(x_i(t))_i\in \R^{n\times d}$ describes the information stacked locally at each node ($x_i(t)$ being the vector at node $i$), we perform the operation $x(t+1)=Wx(t)$. It is to be noted that, thanks to the sparsity of the gossip matrix, this operation is local: for all node $i$,
\begin{equation}
    x_i(t+1)=\sum_{j\sim i}W_{ij}x_j(t), 
\end{equation}where $i\sim j$ if they are neighbors or if $i=j$. The convergence bound will be stated below. Intuitively, at each iteration, each node $i$ sends a proportion of its mass to each one of its neighbour, the condition $\sum_{j\sim i} W_{ij}=1$ being the mass conservation.
\begin{prop}[Synchronous Gossip] Let $\gamma_W$ be the eigengap of the laplacian of $G$ weighted by $1-W_{ij}$ at each edge. Then, for all $k=0,1,2...$:
\begin{equation}
    \NRM{x(k)-\Bar{c}}\leq (1-\gamma_W)^k\|c-\Bar{c}\|,
\end{equation} 
where $x(0)=c$, and $\Bar{c}$ is when consensus is reached
\end{prop}
\begin{proof} For $k\geq 0$,
\begin{align*}
    &x(k+1)-\Bar{c}=W(x(k)-\Bar{c})\\
    &\implies \NRM{x(k+1)-\Bar{c}}\leq \lambda_2(W)\NRM{x(k)-\Bar{c}},
\end{align*}
where $\lambda_2$ is the second largest eigenvalue of $W$, $1$ being the largest ($W$ is stochastic symmetric), and $\Bar{c}$ being in the corresponding eigenspace. We conclude by saying that $\lambda_2(W)=1-\gamma_W$ where $\gamma_W$ is the smallest non null eigenvalue of $Id-W$. Notice that $Id-W$ is the laplacian of the graph weighted by $\nu_{ij}=1-W_{ij}$.
\end{proof}
\noindent Then, since every iteration takes a time $\tau_{max}$, denoting time in a continuous way by $t\in \R^+$, we have:
\begin{equation}
    \NRM{x(t)-\Bar{c}}\leq (1-\gamma_W)^{t/\tau_{max}-1}\leq \exp\left(-\frac{\gamma_W}{\tau_{max}}(t-\tau_{max})\right),
\end{equation}
and $\gamma_W/\tau_{\max}\le \gamma_{synch} $ where $\gamma_{synch}$ is the smallest non-null eigenvalue of the laplacian of the graph with weights $\nu_{ij}=\tau_{\max}$.
\subsection{Asynchronous Gossip\label{app:asynch_gossip}} Time is indexed in a continuous way, by $\R^+$. For every edge $e=(ij)\in E$, let $\mathcal{P}_{ij}$ be a Poisson point process (P.p.p.) of constant intensity $p_{ij}>0$ that we will call "clocks", all independent from each other. Updates will be ruled by these processes: at every clock ticking of $\mathcal{P}_{ij}$, nodes $i$ and $j$ update the value they stack by the mean $\frac{x_i+x_j}{2}$. If we write $\mathcal{P}=\bigcup_{ij\in E} \mathcal{P}_{ij}$, $\mathcal{P}$ is a P.p.p. of intensity $I:=\sum_{ij\in E}p_{ij}$. 
\begin{prop}[Asynchronous Continuous Time Bound] Let $(x_t(i))_i$ be the vector stacked on the graph, and $\Bar{c}=(\frac{1}{n}\sum_i c_i,...,\frac{1}{n}\sum_i c_i)^\top$ the consensus, where $c_i=x_i(0)$. Let $\sigma_{asynch}$ be the smallest non null eigenvalue of the laplacian of the graph, weighted by the $p_{ij}$'s. For $t\geq0$, we have: $$\E[\NRM{x(t)-\Bar{c}}^2]\leq \exp(-t\sigma_{asynch})\NRM{c-\Bar{c}}^2.$$
\end{prop}
\begin{proof}
First, it is to be noted that, if $\mathcal{P}$ is a \emph{P.p.p.} of intensity $\lambda>0$, for all $t\in \R$ and $dt\to 0$: \begin{equation}
    \P([t,t+dt]\cap \mathcal{P}\ne \emptyset)=\lambda dt +o(dt).
\end{equation}
When $ij$ activated at time $t$, multiply $x(t)$ by $W_{ij}=I_n-\frac{^t(e_i-e_j)(e_i-e_j)}{2}$. By observing that $W_{ij}^2=W_{ij}$ and that $\sum_{ij}p_{ij}W_{ij}=I I_n-L$, where $L$ is the laplacian of the graph weighted by the $p_{ij}$, we get that, with $R_t^2=\NRM{x(t)-\Bar{c}}^2$ the squared error to the consensus at time $t$, up to a $o(dt)$:
\begin{align*}
    \E^{\F_t}[R_{t+dt}^2]= & (1-Idt)\E^{\F_t}\left[R_{t+dt}^2|\text{no activations in }[t,t+dt]\right]\\
    & +dt\sum_{ij}p_{ij}\E^{\F_t}\left[R_{t+dt}^2|ij\text{ activated in }[t,t+dt]\right] +o(dt)\\
    & =R_t^2 - dt(x(t)-\Bar{c})^\top\sum_{ij} W_{ij} (x(t)-\Bar{c})\\
    & \leq R_t^2 - dt\sigma_p R_t^2.
\end{align*}
Then, taking the mean, dividing by $dt\to0$ and integrating concudes the proof.
\end{proof}
\subsection{Laplacian Monotonicity} We finish by proving the following intuitive result:
\begin{prop}[Monotonicity of the Laplacian] Let $\Lambda(\lambda_{ij},(ij)\in E)$ be the laplacian of the graph weighted by $\lambda_{ij}$. Then, its second smallest eigenvalue $\sigma$ is a non decreasing function of each weight $\lambda_{ij}$.
\end{prop}
\begin{proof}
First compute $\langle \Lambda u,u\rangle $, the weights $\lambda_{ij}$ being fixed:
\begin{align*}
    \langle \Lambda u,u\rangle  &= \sum_i \sum_{j\sim i} u_i(u_i-u_j) \lambda_{ij}\\
    &=\frac{1}{2} \sum_i \sum_{j\sim i} (u_i-u_j)^2 \lambda_{ij}.
\end{align*}
It appears that for any $u\in \R^n$, these are non decreasing quantities in each $\lambda_{ij}$. If we take $\Lambda$ and $\Lambda'$ two laplacians with weights $\lambda_{ij}\leq \lambda_{ij}'$, we get, for all $u\in \R^n$, $\langle \Lambda u,u\rangle \leq \langle \Lambda' u,u\rangle $. Then, using that $\sigma=\min_{\|u\|=1,\langle u,\mathbb{I}\rangle =0} \langle \Lambda u,u\rangle $ (as $\mathbb{I}$ is a eigenvector associated to the eigenvalue $0$), we have $\sigma'\leq \sigma$ the desired result.
\end{proof}

\section{Preliminary Inequalities \label{app:prelim}}We first present preliminary inequalities using properties on our function $F_A^*$. These properties were also proven in~\citet{hendrikx2018accelerated} (except for Lemma~\ref{lemma:grad_domination}) but we present them here for the paper to be self-contained.

\begin{lem}\label{lem:primal_dual}
Let $x,v\in\R^{n\times d}$ such that $v=\nabla F(x)$ is the dual conjugate. Assume that there exists $\lambda\in\R^{E\times d}$ such that $A\lambda=v$. Let $v^\star$ be the minimizer of $F^*$ on ${\rm Im}(A)={\rm Vect}((1,...,1)^\top)$, $x^\star$ the minimizer of $F$ under consensus constraint and $\lambda^\star$ a minimizer of $F_A^*$. We have:
\begin{equation}
    \NRM{x-x^\star}^2\le \frac{2 L_{\max}}{\sigma_{\min}^2}(F^*(v)-F^*(v^\star)).
\end{equation}
\end{lem}
\begin{proof}
\begin{align*}
    \NRM{x-x^\star}^2&=\NRM{\nabla F^*(v)-\nabla F^*(v^\star)}^2\\
    &\le \frac{1}{\sigma_{\min}^2}\NRM{v-v^\star}^2\text{ (smoothness of $F^*$)}\\
    &= \frac{1}{\sigma_{\min}^2}\NRM{v-v^\star}_{{\rm Im}(A)}^2\\
    &\le \frac{2L_{\max}}{\sigma_{\min}^2}(F^*(v)-F^*(v^\star)) \text{ (strong convexity of $F^*$)}.
\end{align*}
\end{proof}

\begin{lem}\label{lemma:smoothness} For $\lambda\in \R^{E\times d}$ and $ij\in E$, we have:
\begin{equation}
    F_A^*\left(\lambda-\frac{1}{\mu_{ij}^2(\sigma_i^{-1}+\sigma_j^{-1})}U_{ij}\nabla_{ij} F_A^*(\lambda)\right)-F_A^*(\lambda)\leq -\frac{1}{2\mu_{ij}^2(\sigma_i^{-1}+\sigma_j^{-1})}\|\nabla_{ij}F_A^*(\lambda)\|^2.
    \label{smoothness}
\end{equation} 
\end{lem}
\begin{proof}
Let us define $h_{ij}=-\frac{1}{\mu_{ij}^2(\sigma_i^{-1}+\sigma_j^{-1})}U_{ij}\nabla_{ij} F_A^*(\lambda)$.
\begin{align*}
    F_A^*\left(\lambda+h_{ij}\right)-F_A^*(\lambda)&=\sum_k f^*_k((A\lambda)_k+(Ah_{ij})_k))-f_k^*((A\lambda)_k)\\
    &= f_i^*((A\lambda)_i+(Ah_{ij})_i)-f_i^*((A\lambda)_i)+f_j^*((A\lambda)_j+(Ah_{ij})_j)-f_j^*((A\lambda)_j),
\end{align*}
as $(Ah_{ij})$ is supported only by coordinates $i$ and $j$. Moreover, as $f_i^*$ is $\sigma_i$-smooth, we have:
\begin{align*}
    f_i^*((A\lambda)_i+(Ah_{ij})_i)-f_i^*((A\lambda)_i)\leq \langle  \nabla f_i^*((A\lambda)_i),(Ah_{ij})_i\rangle  +\frac{\sigma_i^{-1}}{2} \|(Ah_{ij})_i\|^2,
\end{align*}
and by summing for $i$ and $j$ and noticing that $(Ah_{ij})_i=\mu_{ij}\nabla_{ij}F_A^*(\lambda)$:
\begin{align*}
    F_A^*(\lambda+h_{ij})-F_A^*(\lambda)&\leq\langle  \nabla_{ij} F_A(\lambda),h_{ij}\rangle  +\frac{(\sigma_i^{-1}+\sigma_j^{-1})\mu_{ij}^2}{2} \left(\frac{1}{\mu_{ij}^2(\sigma_i^{-1}+\sigma_j^{-1})}\right)^2\|\nabla_{ij} F_A^*(\lambda)\|^2\\
    &=-\frac{1}{2\mu_{ij}^2(\sigma_i^{-1}+\sigma_j^{-1})}\|\nabla_{ij}F_A^*(\lambda)\|^2.
\end{align*}
\end{proof}
\begin{lem}\label{lemma:sc}
$\sigma_A$ the strong convexity parameter of $F_A^*$ on the orthogonal of $Ker(A)$ is lower bounded by $\lambda_{min}^+(A^TA)/L_{max}$, where $\lambda_{min}^+(A^TA)$ is the smallest non null eigenvalue of $A^TA$.
\end{lem}
\begin{proof}
Let $\lambda,\lambda'\in \R^{E\times d}$. By $L_i^{-1}$ and thus $L_{max}^{-1}$-strong convexity of $f_i^*$:
\begin{align*}
    f_i^*((A\lambda)_i)-f_i^*((A\lambda')_j)&\geq \langle  \nabla f_i^*((A\lambda')_i),(A(\lambda-\lambda'))_i\rangle  -\frac{1}{2L_{max}}\|(A(\lambda-\lambda')\|^2\\
\end{align*}
Summing over all $i\in [n]$ and using $\nabla F_A^*(\lambda') = ^tA (\nabla_i f_i^*((A\lambda')_i))_i$ leads to:
\begin{align*}
    F_A^*(\lambda)-F_A^*(\lambda')&\geq \langle  \nabla F_A^*(\lambda'),\lambda-\lambda'\rangle  -\frac{1}{2L_{\max}}\|A(\lambda'-\lambda\|^2\\
    &\geq \langle  \nabla F_A^*(\lambda'),\lambda-\lambda'\rangle  -\frac{\lambda^+_{\min}(A^TA)}{2L_{max}}{\|\lambda-\lambda'\|^*}^2.
\end{align*}
where $\|.\|^*$ is the euclidian norm on the orthogonal of $Ker(A)$.
\end{proof}
\begin{lem}\label{lemma:AAT}
$AA^T$ is the laplacian of the graph $G$ weighted by $\mu_{ij}^2$ on the edges.
\end{lem}
\begin{proof}
\begin{align*}
    A^T e_i=\sum_{j\sim i}\mu_{ij}e_{ij}
\end{align*}
\noindent For the diagonal, we have:
\begin{align*}
    e_i A A^T e_i&=\sum_{k\sim i}\sum_{l\sim i} \mu_{ik}\mu_{il}\langle  e_{ik},e_{il}\rangle  \\
    &= \sum_{j\sim i}\mu_{ij}^2.
\end{align*}
Then, for $i\sim j, i\ne j$:
\begin{align*}
    e_i A A^T e_j&=\sum_{k\sim i}\sum_{l\sim i} \mu_{ik}\mu_{jl}\langle  e_{ik},e_{jl}\rangle  \\
    &=\mu_{ij}\mu_{ji}\\
    &=-\mu_{ij}^2.
\end{align*}
\end{proof}

\begin{lem}
For $x,x'\in R^{E\times d}$, and $ij \in E$, we have:
\begin{equation}
    \|\nabla_{ij}F_A^*(x)-\nabla_{ij}F_A^*(x')\|^2\leq 2(\sigma_i^{-1}+\sigma_j^{-1})^2 d_{ij} \mu_{ij}^2 \sum_{(kl)\sim (ij)}\mu_{kl}^2 \|x_{kl}-x'_{kl}\|^2.
     \label{gossipnoniid1}
\end{equation}
\end{lem}
\begin{proof}
First, notice that $\nabla_{ij}F_A^*(x)=\mu_{ij}(\nabla f_i^*((Ax)_i)-\nabla f_j^*((Ax)_j))$. Then:
\begin{align*}
    \|\nabla f_i^*((Ax)_i)-\nabla f_i^*((Ax')_j)\| &\leq \sigma_i^{-1}\|(A(x-x'))_i\| \text{ (smoothness)}\\
    & \leq \sigma_i^{-1} \|\sum_{kl\sim ij} \mu_{kl} (x-x')_{kl}\|\\
    & \leq \sigma_i^{-1} \sum_{kl\sim ij} \mu_{kl} \|(x-x')_{kl}\|
\end{align*}
Conclude by taking the square and summing for $i$ and $j$.
\end{proof}

\begin{lem}[Distance to Optimum]\label{lemma:grad_domination} For any $\lambda\in \R^{E\times d}$ and for $\lambda^\star $ minimizing $F_A^*$, we have:
\begin{equation}
  \label{opt2strg} F_A^*(\lambda)-F_A^*(\lambda^\star )\leq \frac{1}{2\sigma_A} \|\nabla F_A^*(\lambda)\|^2
\end{equation}
\end{lem}
\begin{proof}
We introduce Bregman divergences, which make the proof straightforward. For $\phi$ any real-valued function, differentiable, defined on an euclidian space $\mathcal{V}$, we define its Bregman divergence $D_{\phi}$ on $\mathcal{V}^2$ by:
\begin{equation}
    D_{\phi}(x,y)=\phi(x)-\phi(y)-\langle  \nabla \phi(y),x-y\rangle  .
\end{equation}
$\phi$ is thus $L$-smooth if and only if $D_{\phi}\leq L D_{\|.\|^2/2}$. An important equality is the following, under convexity assumption for $\phi$:
\begin{equation}
    \label{bregprop}
    D_{\phi}(x,y)=D_{\phi^*}(\nabla \phi (y), \nabla \phi(x)).
\end{equation}
Applying this to $\phi=F_A^*$, $x=\lambda,y=\lambda^\star $, together with the fact that $(F_A^*)^*$ is $\sigma_A^{-1}$-smooth with respect to ${\|.\|^*}^2$~\citep{kakade2009duality}, the squared norm on the orthogonal of $Ker(A)$ leads to:
\begin{equation*}
    D_{F_A^*}(\lambda, \lambda^\star ) = D_{{F_A^*}^*}(\nabla F_A^*(\lambda^\star ), \nabla F_A^*(\lambda)) \leq \frac{1}{\sigma_A}D_{{\|.\|^*}^2/2}(\nabla F_A^*(\lambda^\star ), \nabla F_A^*(\lambda)),
\end{equation*}
and the result follows since $\nabla F_A^*(\lambda^\star ) = 0$ and ${\|\nabla F_A^*(\lambda)\|^*}^2=\|\nabla F_A^*(\lambda)\|^2$.
\end{proof}

\section{Detailed Proof of Theorem \ref{thm:LN}\label{app:LN}}

\subsection{Proof Of Theorem \ref{thmdet}\label{app:general_activations}}

To prove this intermediate theorem, we need to study every gradient step involved. At iteration $s$, not every coordinates is available, hence the need to study the impact of $T$ gradient steps together. A gradient step alongside edge $ij$ only involves edges in its neighborhood (thanks to the sparsity of the matrix $A$), a key element that will need to be explicited. The proof involves three main steps.\\

\noindent \textbf{Step 1:} Applying Lemma~\ref{lemma:smoothness} (local smoothness) gives, where $ij$ is the $t^{th}$ activated edge: 
\begin{equation}
    F_A^*(\lambda(t+1))-F_A^*(\lambda(t))\leq -\frac{1}{2(\sigma_i^{-1}+\sigma_j^{-1})\mu_{ij}^2}\|\nabla_{ij}F_A^*(\lambda(t))\|^2.
\end{equation} 
Hence, we get an inequality between $L_t$ and $L_{t+1}$: 
\begin{equation}
    \Lambda_{t+1}=\frac{1}{T}\sum_{t\leq s < t+T} (F_A^*(\lambda(s+1))-F_A^*(\lambda^\star ))\leq \Lambda_t - \frac{1}{T}\sum_{t\leq s < t+T} \frac{1}{2(\sigma_i^{-1}+\sigma_j^{-1})\mu_{(ij)_s}^2}\|\nabla_{(ij)_s}F_A^*(\lambda(s))\|^2
\end{equation}
where $(ij)_s$ is the edge activated during activation $s$. Let's introduce the following quantity: \begin{equation}
    \frac{1}{T}\sum_{t\leq s < t+T} \sum_{ij\in E} \|\nabla_{ij}F_A^*(\lambda(s))\|^2 = \frac{1}{T}\sum_{t\leq s < t+T}\|\nabla F_A^*(\lambda(s))\|^2 \geq \sigma_A \Lambda_t
\end{equation} 
where where we used Lemma~\ref{lemma:grad_domination} (gradient domination), and $ \sigma_A$ is the strong convexity parameter of $F_A^*$ (lower bounded by $\lambda_{min}^+(A^TA)/L_{max}$). Hence, if an inequality of the type
\begin{equation}\label{neededthm2}
    \frac{C}{T}\sum_{t\leq s < t+T} \sum_{ij\in E} \|\nabla_{ij}F_A^*(\lambda(s))\|^2\leq \frac{1}{T}\sum_{t\leq s < t+T} \frac{1}{2(\sigma_i^{-1}+\sigma_j^{-1})\mu_{(ij)_s}^2}\|\nabla_{(ij)_s}F_A^*(\lambda(s))\|^2
\end{equation}
holds, we have (using \eqref{opt2strg}):
\begin{equation}
    \Lambda_{t+1}\leq L_t-C\frac{1}{T}\sum_{t\leq s < t+T} \|\nabla F_A^*(\lambda(s))\|^2\leq (1-C\sigma_A)\Lambda_t.
\end{equation}
We thus need to tune correctly the $\mu_{ij}^2$ and $C$ in order to have \eqref{neededthm2} verified.\\

\noindent \textbf{Step 2:} We are looking for necessary conditions for \eqref{neededthm2} to hold. In the left term, every coordinate is present at each time $s$. However, in the right hand side of the inequality, just the activated one is present. We will need to compensate this with a bigger factor in front of the gradients. In order to compare these quantities, we need to introduce upper bound inequalities on $\|\nabla_{ij}F_A^*(\lambda(s))\|^2$, that only make activated coordinates intervene. Let $s\in \{t,...,t+T-1\}$, and suppose that there exists $t\leq r \leq s < r+t_{ij} \leq t+T-1$ such that $ij$ is activated at times $r$ and $r+t_{ij}$. Thanks to the asumption on $T$, either one of these integers exists. If the other one doesn't, replace it with $t$ for $r$, and by $t+T-1$ for $r+t_{ij}$. Thanks to our asumptions, we know that $t_{ij}\leq a\ell_{ij}$. We have the following basic inequalities:
\begin{align}
    \|\nabla_{ij}F_A^*(\lambda(s))\|^2 &\leq (\|\nabla_{ij}F_A^*(\lambda(r))\|+\|\nabla_{ij}F_A^*(\lambda(s))-\nabla_{ij}F_A^*(\lambda(r))\|)^2\\
    &\leq 2(\|\nabla_{ij}F_A^*(\lambda(r))\|^2+\|\nabla_{ij}F_A^*(\lambda(s))-\nabla_{ij}F_A^*(\lambda(r))\|^2).
\end{align}
The quantity $\|\nabla_{ij}F_A^*(\lambda(s))-\nabla_{ij}F_A^*(\lambda(r))\|^2$ then needs to be controlled. We know that thanks to \eqref{gossipnoniid1}, for $x,x'\in \R^{E\times d}$, we have 
\begin{equation}
    \|\nabla_{ij}F_A^*(x)-\nabla_{ij}F_A^*(x')\|^2\leq 2(\sigma_i^{-1}+\sigma_j^{-1})^2 d_{ij} \mu_{ij}^2 \sum_{(kl)\sim (ij)}\mu_{kl}^2 \|x_{kl}-x'_{kl}\|^2.
\end{equation}
Using this with 
\begin{align}
    \|x_{kl}-x'_{kl}\|^2 & =\|\sum_{r<u<s:(ij)_u=(kl)}\frac{1}{(\sigma_k^{-1}+\sigma_l^{-1})\mu_{kl}^2}\nabla_{kl}F_A^*(\lambda(u))\|^2\\
    & \leq \sum_{r<u<r+t_{ij}:(ij)_u=(kl)}\left(\frac{1}{(\sigma_k^{-1}+\sigma_l^{-1})\mu_{kl}^2}\right)^2N(kl,ij,u)\|\nabla_{kl}F_A^*(\lambda(u))\|^2, 
\end{align}
where we used (and will widely use again below) that $\|x_1+...+x_n\|^2\leq n (\|x_1\|^2+...+\|x_n\|^2)$ (convexity of the squared norm), leads to:
\begin{align}
    \|\nabla_{ij}F_A^*(\lambda(s))\|^2& \leq 2\|\nabla_{ij}F_A^*(\lambda(r))\|^2\\
    &+2d_{ij}\sum_{r<u<r+t_{ij}}N((ij)_u,ij,u)\frac{\mu_{ij}^2(\sigma_i^{-1}+\sigma_j^{-1})^2}{\mu_{(ij)_u}^2(\sigma_{i_u}^{-1}+\sigma_{j_u}^{-1})^2}\|\nabla_{(ij)_u}F_A^*(\lambda(u))\|^2\\
    & \leq 2\|\nabla_{ij}F_A^*(\lambda(r))\|^2\\
    &+2d_{ij}\sum_{r<u<r+t_{ij}} \left\lceil b\frac{\ell_{ij}}{L_{(ij)_u}}\right\rceil\frac{\mu_{ij}^2(\sigma_i^{-1}+\sigma_j^{-1})^2}{\mu_{(ij)_u}^2(\sigma_{i_u}^{-1}+\sigma_{j_u}^{-1})^2}\|\nabla_{(ij)_u}F_A^*(\lambda(u))\|^2
\end{align}
The advantage of this last expression is that only activated quantities are present on the right hand side.\\

\noindent \textbf{Step 3:} The last step of the proof consists in summing the last inequality for $t\leq s <t+T$, $ij\in E$. When summing, each $\|\nabla_{(ij)_r}F_A^*(\lambda(r))\|^2$ appears on the right hand-side of the inequality, with a factor upper-bounded by ($(ij)_r$ noted $(ij)$):
\begin{equation}
    2a\ell_{ij}+2 d_{ij}\sum_{kl\sim ij}a\ell_{kl}\left\lceil\frac{b\ell_{kl}}{\ell_{ij}}\right\rceil\frac{\mu_{kl}^2(\sigma_{k}^{-1}+\sigma_{l}^{-1})^2}{\mu_{ij}^2(\sigma_i^{-1}+\sigma_j^{-1})^2}.
\end{equation}
We want the expression above multiplied by $C$ defined in Step 1 to be upper-bounded by $\frac{1}{2(\sigma_i^{-1}+\sigma_j^{-1})\mu_{ij}^2}$, in order for \eqref{neededthm2} to be verified. This is possible if and only if:
\begin{equation}
     C\left(2a\ell_{ij}\mu_{ij}^2(\sigma_i^{-1}+\sigma_j^{-1})+2 d_{ij}\sum_{kl\sim ij}a\left\lceil\frac{b\ell_{kl}}{\ell_{ij}}\right\rceil \ell_{kl}\mu_{kl}^2\frac{(\sigma_{k}^{-1}+\sigma_{l}^{-1})^2}{\sigma_{i}^{-1}+\sigma_{j}^{-1}}\right)\leq \frac{1}{2},
     \label{thm1last}
\end{equation}
where $C$ is defined in step $1$ of the proof. This is equivalent to:
\begin{align}
      C\left(2a \ell_{ij}\mu_{ij}^2(\sigma_i^{-1}+\sigma_j^{-1})+2 d_{ij}\sum_{kl\sim ij}a\frac{b \ell_{kl}^2}{ \ell_{ij}}\mu_{kl}^2\frac{(\sigma_{k}^{-1}+\sigma_{l}^{-1})^2}{\sigma_{i}^{-1}+\sigma_{j}^{-1}}\right)\leq \frac{1}{4} \text{ if } \forall kl\sim ij, \ell_{ij}\leq b \ell_{kl},
\end{align}
where we bounded $\left\lceil b\frac{ \ell_{ij}}{ \ell_{kl}}\right\rceil$ by $2\frac{b \ell_{ij}}{ \ell_{kl}}$ here. We here see that in this case, if 
\begin{equation}
    \mu_{ij}^2=\frac{1}{ \ell_{ij}(\sigma_i^{-1}+\sigma_j^{-1})}\times \min_{kl\sim ij}\frac{ \ell_{kl}(\sigma_k^{-1}+\sigma_l^{-1})}{ \ell_{ij}(\sigma_i^{-1}+\sigma_j^{-1})}
\end{equation}
with $8a+8d_{max}^2b\leq C^{-1}$, our inequality holds. However, our inequality on the ceil operator seems not to work in the general case. Let's take $kl$ a neighbor of $ij$ such that $ \ell_{ij}>b \ell_{kl}$. As $ \ell_{ij}>b \ell_{kl}$, we have $\lceil\frac{b \ell_{kl}}{ \ell_{ij}}\rceil=1$, leading to $a\lceil\frac{b \ell_{kl}}{ \ell_{ij}}\rceil \ell_{kl}\mu_{kl}^2=a \ell_{kl}\mu_{kl}^2\leq a\leq ab$. Hence, our result still holds.\\

\noindent \textbf{Conclusion:} We have our result for $C=\frac{1}{2a+8d_{max}^2ab}$ and a laplacian weighted with local communication constraints: $\mu_{ij}^2=\frac{1}{ \ell_{ij}(\sigma_i^{-1}+\sigma_j^{-1})}\times \min_{kl\sim ij}\frac{ \ell_{kl}(\sigma_k^{-1}+\sigma_l^{-1})}{ \ell_{ij}(\sigma_i^{-1}+\sigma_j^{-1})}$. The final rate thus depends on the smallest eigenvalue of the laplacian weighted by:
\begin{equation}
    \frac{1}{2a+8d_{max}^2ab}\frac{1}{L_{max}}\frac{1}{ \ell_{ij}(\sigma_i^{-1}+\sigma_j^{-1})}\times \min_{kl\sim ij}\frac{ \ell_{kl}(\sigma_k^{-1}+\sigma_l^{-1})}{ \ell_{ij}(\sigma_i^{-1}+\sigma_j^{-1})}.
\end{equation}
However, having local complexity constraints is not really of much interest to us, as the parameters $\sigma_i$ entered in the algorithm are generally taken to be the same on all nodes. We thus formulate Theorem 2 with $\sigma_{min}$ for simplicity (which is slightly weaker in general) which gives as final rate of convergence the smallest eigenvalue of the laplacian weighted by:
\begin{equation}
    \nu_{ij}=\frac{1}{2a+8d_{\max}^2ab}\frac{\sigma_{\min}}{2L_{\max}}\frac{1}{ \ell_{ij}}\times \min_{kl\sim ij}\frac{ \ell_{kl}}{ \ell_{ij}}.
\end{equation}

\subsection{Proof Of Proposition \ref{thmsto}: Adding Stochasticity \label{app:LN_sto}}

We now prove the other theorem, where we assume the existence of events $A_t$ for $t\in \N$, under which the asumptions are true. Using the same arguments as in the proof of Theorem 2, we obtain:
\begin{equation}
    \E[  \Lambda_{t+1}-  \Lambda_t|\F_t,A_t]\leq -\sigma   \Lambda_t.
\end{equation}
However, this is not enough to conclude. Under $A_t^C$, we only know that $  \Lambda_{t+1}\leq   \Lambda_t$ using Lemma~\ref{lemma:smoothness} (our local gradient steps cannot increase distance to the optimum). Hence:
\begin{equation}
    \E[  \Lambda_{t+1}|\F_t]\leq (1-\sigma \mathbb{I}_{A_t})  \Lambda_t.
\end{equation}
And then, by induction:
\begin{equation}
    \E[  \Lambda_t]\leq \E[P_t\Lambda_0]  ,\text{ where } P_t=\prod_{s=0}^{t-1}(1-\sigma \mathbb{I}_{A_s}).
\end{equation}
However, no direct bound on $P_t$ exists. The interdependencies on the events $A_t$ make it impossible for an induction to prove a bound of the form $\leq (1-\sigma/2)^t$. However, the logarithm of the product seems easier to study:
\begin{equation}
    \log(P_t)=\log(1-\sigma)\sum_{s=0}^{t-1}\mathbb{I}_{A_s},
\end{equation}
giving us $\E\log(P_t)\leq\log(1-\sigma)t/2$, as $\P(A_t)\geq 1/2$. We are thus going to make a study in probability. For $t\in \N$, let $X_t=\frac{1}{T}\sum_{s=t}^{t+T-1}\mathbb{I}_{A_s}$. Using Markov-type inequalities conditionnaly on $\F_t$ gives:
\begin{equation}
    \P(X_t\geq 1/3|\F_t) +1/3\P(X_t\leq 1/3|\F_t) \geq \E[ X_t|\F_t] \geq 1/2 \implies \P(X_t\geq 1/3|\F_t)\geq 1/4.
\end{equation}
Thus, we have: $\E[\prod_{s=t}^{t+T-1}(1-\mathbb{I}_{A_s}\sigma)|\F_t]\leq \frac{1}{4}(1-\sigma)^{T/3}+\frac{3}{4}.$ We then know how to control $T$ consecutive factors of the product $P_t$. Skipping the next $T$ terms, we have:
\begin{align}
    \E\left[\prod_{s=t}^{t+3T-1}(1-\mathbb{I}_{A_s}\sigma)\right] & =\E\left[\prod_{s=t}^{t+T-1}(1-\mathbb{I}_{A_s}\sigma)\prod_{s=t+T}^{t+2T-1}(1-\mathbb{I}_{A_s}\sigma)\prod_{s=t+2T}^{t+3T-1}(1-\mathbb{I}_{A_s}\sigma)\right]\\
    & \leq \E\left[\prod_{s=t}^{t+T-1}(1-\mathbb{I}_{A_s}\sigma)\prod_{s=t+2T}^{t+3T-1}(1-\mathbb{I}_{A_s}\sigma)\right]\\
    & \leq \E\left[\prod_{s=t}^{t+T-1}(1-\mathbb{I}_{A_s}\sigma)\E^{\F_{t+2T}}\left \{\prod_{s=t+2T}^{t+3T-1}(1-\mathbb{I}_{A_s}\sigma)\right\} \right]
\end{align}
as in the last right hand side, the first big product is $\F_{t+2T}$-measurable (our asumption on the $A_s$ states that they are  $\F_{s+T-1}$-measurable). Then, using inequality $\E\left[\prod_{s=t}^{t+T-1}(1-\mathbb{I}_{A_s}\sigma)|\F_t\right]\leq \frac{1}{4}(1-\sigma)^{T/3}+\frac{3}{4}$ twice, with $t$ and $t+2T$, we get:
\begin{align*}
\E\left[\prod_{s=t}^{t+3T-1}(1-\mathbb{I}_{A_s}\sigma)\right]&\leq \E\left[\prod_{s=t}^{t+T-1}(1-\mathbb{I}_{A_s}\sigma) \left(\frac{1}{4}(1-\sigma)^{T/3}+\frac{3}{4}\right)\right]\\
&\leq\left(\frac{1}{4}(1-\sigma)^{T/3}+\frac{3}{4}\right)^2.\end{align*} Proceeding the same way by induction leads us to: 
\begin{equation}
    \E[P_t]\leq \left(\frac{1}{4}(1-\sigma)^{T/3}+\frac{3}{4}\right)^{\lfloor t/(2T) \rfloor},
\end{equation}
which is the desired bound. For the asymptotic one, $(1-\sigma)^{T/3}\leq e^{-\sigma T/3}$. For $\sigma T$ small enough (less than $\log(2)$), we have $e^{-\sigma T/3}\leq 1-\sigma T/3$, leading to $(\frac{1}{4}(1-\sigma)^{T/3}+\frac{3}{4})^{\lfloor t/(2T)}\leq (1-T\sigma/12)^{\lfloor t/(2T)}\leq e^{-(t+o(t))\sigma/24}.$ The asymptotic rate of convergence thus holds if the assumption made in Corollary 1 holds.

\subsection{Study in the \emph{RLNM($\eps$)}: Tuning the Parameters\label{app:LN_tuning}}

We first assume to be in the case $\eps=0$. We generalize to $\eps>0$ at the end.
Let $t\in \N$ be fixed, and $B_t$ be the event: "in the activations $t,t+1,...,t+T-1$, all edges are ativated". Let then $C_t(ij,s)$ for $t\leq s < t+T$ be the event $\min(T_{ij}(s),t+T-s,s-t)\leq a\ell_{ij}$ and $D_t(kl,ij,s)$ be the event $N(kl,ij,s)\leq \lceil b\ell_{ij}/\ell_{kl} \rceil$, where $N(kl,ij,s)$ is the number of activations of $kl$ between two activations of $ij$, around time $s$, where we only take into account the activations between times $t$ and $t+T-1$. Let then $A_t=B_t\cap (\cap_{kl,ij\in E, t\leq s < t+T}C_t(ij,s)\cap D_t(kl,ij,s))$. We want $\P(A_t)\geq 1/2$ for correct constants $a,b,T$ and $\ell_{ij}$ (that can differ from $\tau_{ij}$). Note that this event is $\F_{t+T-1}$-measurable, as desired. We first study the length of time $\ell_{ij}$ edge $ij$ must wait in order to be activated with high probability (\emph{high} meaning more that $1-\frac{1}{12|E|}$). This result is Lemma \ref{lem_queue_1}. Then, we use this length to determine the constants $T,a,b,\ell_{ij}$ needed.\\

\begin{lem}
For any $t_0\geq 0$, $ij\in E$, if $p_{ij}=\frac{1}{2\max(d_i,d_j)-1}\tau_{ij}^{-1}$ and $\tau_{max}(ij)=\max_{kl\sim ij}\tau_{kl}$, let $\ell_{ij}=\frac{\log(6|E|)}{\log(1-(1-e^{-1})e^{-1})}(p_{ij}^{-1}+\tau_{max}(ij))$. We have:
\begin{equation}
    \P(ij\text{ not activated in $[t_0,t_0+\ell_{ij}]$}|\F_{t_0})\leq \frac{1}{6|E|}.
\end{equation}
\end{lem}
\begin{proof}[Proof of Lemma \ref{lem_queue_1}]Let $ij\in E$ and $t_0\geq 0$ fixed. We use tools from queuing theory \citep{Tanner1995queuing} ($M/M/\infty/\infty$ queues) in order to compute the probability that edge $ij$ is activable at a time $t$ or not. More formally, we define a process $N_{ij}(t)$ with values in $\N$, such that $N_{ij}(t_0)=1$ if $ij$ non-available at time $t_0$ and $0$ otherwise. Then, when an edge $kl,kl\sim ij$ is activated, we make an increment of $1$ on $N_{ij}(t)$ (a \emph{customer} arrives). This customer stays for a time $\tau_{kl}$ and when he leaves we make $N_{ij}$ decrease by $1$. We have $N_{ij}\geq 0$ a.s., and if $N_{ij}=0$, $ij$ is available. For $t\geq \max_{kl\sim ij} \tau_{kl} +t_0$, $N_{ij}(t)$ follows a Poisson law of parameter $\sum_{kl\sim ij}p_{kl}\tau_{kl}$.  For any $t\geq \max_{kl\sim ij} \tau_{kl} +t_0$:
\begin{equation}
    \P(ij\text{ available at time }t|\F_{t_0})\geq \P(N_i(t)= 0)=\exp(-\sum_{kl\sim ij}p_{kl}\tau_{kl}).
\end{equation}
That leads to taking $p_{kl}=\frac{1}{2}\frac{1}{\max(d_k,d_l)-1}\tau_{kl}^{-1}$ for all edges, in order to have $\P(ij\text{ available at time }t|\F_{t_0})\geq 1/e$. Then, $\P(ij \text{ rings in } [t,t+p_{ij}^{-1}])=1-e^{-1}$, giving:
\begin{align}
    \P(ij &\text{ activated in }[t_0,t_0+\tau_{\max}(ij)+p_{ij}^{-1}]|\F_{t_0})= \P(ij \text{ rings in } [t,t+p_{ij}^{-1}])\\
    &\times \P(ij \text{ available at time }t|\F_{t_0},\text{$ij$ rings at a time}\\
    & t\in [t_0+\tau_{\max}(ij),t_0+\tau_{\max}(ij)+p_{ij}^{-1}])\\
    & \geq (1-e^{-1})e^{-1},
\end{align}
where we use the fact that exponential random variables have no memory. Take $k\in \N$ such that $(1-(1-e^{-1})e^{-1})^k\leq \frac{1}{6|E|}$, leading to $k\approx \log(6|E|)/\log(1-(1-e^{-1})e^{-1})$. Let $\ell_{ij}=k(p_{ij}^{-1}+\tau_{max}(ij))$. Then we have a.s.:
\begin{equation}
    \P(ij\text{ not activated in $[t_0,t_0+\ell_{ij}]$}|\F_{t_0})\leq \frac{1}{6|E|}.
\end{equation}
\end{proof}

\noindent \textbf{Bounding $T$:} A direct application of Lemma \ref{lem_queue_1} leads, with $L=\max_{ij}\ell_{ij}$, to:
\begin{equation}
    T=2\sum_{ij}\frac{L}{\tau_{ij}}.
\end{equation}
Indeed, for all $ij$, not being activated in activations $t,t+1,...,t+T-1$ means not being activated for a continuous interval of time of length more than $\ell_{ij}$. Hence:
\begin{align}
    &\P(\exists (ij)\in E: (ij)\text{ not activated in }\{t,...,t+T-1\}|\F_t)\\
    &\leq \sum_{ij\in E}\P((ij)\text{ not activated in }\{t,...,t+T-1\}|\F_t)\\
    &\leq \sum_{ij\in E}\P((ij)\text{ not activated in }[t,t+\ell_{ij}]|\F_t)\\
    &\leq |E|\times \frac{1}{6|E|}\\
    &=1/6. \label{eq:B_t}
\end{align}

\noindent \textbf{Bounding $T_{ij}$:} Applying Lemma \ref{lem_queue_1} with $12|E|T$ instead of $6|E|$ leads to controlling all the inactivation lengths by a length $\ell'_{ij}$, with a probability more than $1-1/(12|E|T)$. Let $ij\in E$ and $s\in \N$, $t\leq s < t+T$. Let $\alpha>0$ to tune later. Denote by $\delta_{ij}(s)$ the (random) inactivation time of $ij$, around iteration $s$. Note that conditionnaly on the inactivation period $\delta_{ij}(s)$, $T_{ij}(s)$ is dominated in law by a Poisson variable of parameter $I\delta_{ij}(s)$, hence line \eqref{eq:poissondelta}:
\begin{align}
    \P(T_{ij}(s)\geq \alpha \ell_{ij}'|\F_t)&\leq\P(T_{ij}(s)\geq \alpha \ell_{ij}'|\F_t, \delta_{ij}\leq \ell_{ij}')\times \P(\delta_{ij}\leq \ell_{ij}') + \P(\delta_{ij}\geq \ell_{ij}')\\
    &\leq \P(Poisson(I\ell'_{ij})\geq \alpha \ell_{ij}') + \frac{1}{12|E|T}\label{eq:poissondelta}\\
    &\leq \frac{1}{12|E|T}+\frac{1}{12|E|T}\label{eq:wantedpoisson}\\
    &=\frac{1}{6|E|T},
\end{align}
for some $\alpha>0$ big enough, to determine with the following large deviation inequality:
\begin{lem}[A Large Deviation Inequality on discrete Poisson variables.] \label{poisson}Let $Z\sim Poisson(\lambda)$, for some $\lambda>0$. Then, for all $u\geq 0$:
\begin{equation}
    \P(Z\geq u)\leq \exp(-u+\lambda(e-1)).
\end{equation}
\end{lem}
\noindent This large deviation leads to taking $\alpha=2eI$ for \eqref{eq:wantedpoisson} to be true. Finally, we get:
\begin{equation}
    \P(T_{ij}(s)\geq \alpha \ell_{ij}'|\F_t)\leq \frac{1}{6|E|T}. \label{eq:C_t}
\end{equation}

\noindent \textbf{Bounding $N(kl,ij,s)$:} If $\delta_{ij}(s)\leq \ell'_{ij}$, this random variable is dominated by a Poisson variable of parameter $p_{kl}\ell_{ij}'$. Hence, still with Lemma \ref{poisson}, with probability more than $1-\frac{1}{12|E|^2T}$, we can bound $N(kl,ij)$ by $e\log(12|E|^2T)+p_{kl}\ell_{ij}(e-1)\leq 2 e p_{kl}L_{ij}$.\\

\noindent \textbf{Explicit writing of the union bound on $A_t^C$:} $A_t^C=B_t^C\cup (\cup_{kl,ij\in E, t\leq s < t+T}C_t(ij,s)^C\cup D_t(kl,ij,s)^C)\in \F_{t+T-1}$. Thanks to the previous considerations, we have that $\P^{\F_t}(B_t^C) \leq 1/6$ with \eqref{eq:B_t}, $\P^{\F_t}(C_t(ij,s)^C)\leq \frac{1}{6|E|T}$ with \eqref{eq:C_t} and $\P(D_t(kl,ij,s)^C|\F_t)\leq \frac{1}{6|E|^2T}$, for the following constants and weights:
\begin{itemize}
    \item $\Tilde{\tau}_{ij}^{-1}=p_{ij}=\min(\frac{1}{\tau_{\max}(ij)},\frac{1}{2(\max(d_i,d_j)-1)}\frac{1}{\tau_{ij}})$;
    \item $T=2 I \max_{ij\in E}\Tilde{\tau_{ij}}\frac{\log(6|E|)}{\log(1-(1-e^{-1})e^{-1})}$;
    \item $a=2eI\frac{\log(6|E|T)}{\log(1-(1-e^{-1})e^{-1})}$;
    \item $b=2e\frac{\log(6|E|T)}{\log(1-(1-e^{-1})e^{-1})}$.
\end{itemize}
The union bound is the following:
\begin{align}
    \P^{\F_t}(A_t^C) & \leq  \P^{\F_t}(B_t^C) +\sum_{s,ij}\P^{\F_t}(C_t(ij,s)^C) +\sum_{s,ij}\P^{\F_t}(\cup_{kl}D_t(kl,ij,s)^C)\\
    & \leq 1/6 + |E|T/(6|E|T)\times 2\\
    & \leq 1/2.
\end{align}

\noindent The rate of convergence $\rho$ is then defined as the smallest non null eigenvalue of the laplacian of the graph, weighted by:
\begin{equation}
   \nu_{ij}= \frac{\sigma_{min}}{L_{max}}\times \frac{\Tilde{\tau_{ij}}\min_{kl\sim ij}\frac{\tau_{ij}}{\tau_{kl}}}{8a(1+d^2b)}.
\end{equation}
Note that this analysis works for $\eps=0$, but also for \textbf{RLNM($\eps>0$)} by replacing $\tau_{ij}$ by $(1+\eps)\tau_{ij}$. Indeed, Lemma \ref{lem_queue_1} still holds with $(1+\eps)\tau_{ij}$: the queuing construction still works.

\subsection{Proof of Corollary \ref{cor:LN} \label{app:cor1}}

\begin{proof}
First, notice that $\mathcal{E}_k\le \mathcal{L}_k$ since the sequence $(\mathcal{E}_l)_l$ is non-increasing. Then:
\begin{align*}
    \left(\frac{1}{4}(1-\frac{\sigma_{\min}}{L_{\max}}\Gamma_{RLNM})^{T/3}+\frac{3}{4}\right)^{\lceil \frac{k}{2T}\rceil}& \le \left(\frac{1}{4}\exp(-\frac{\sigma_{\min}}{L_{\max}}\Gamma_{RLNM}\frac{T}{3})+\frac{3}{4}\right)^{\lceil \frac{k}{2T}\rceil}\\
    &\le \left(1-\frac{\sigma_{\min}}{L_{\max}}\Gamma_{RLNM}\frac{T}{12e})\right)^{\lceil \frac{k}{2T}\rceil} \text{ if } \frac{\sigma_{\min}}{L_{\max}}\Gamma_{RLNM}\frac{T}{12}\le 1.
\end{align*}
That last condition is satisfied under Assumption \ref{hyp:delay_constraint} using monotonicity of the Laplacian. We thus have our result taking the logarithm and making $k\to \infty$.
\end{proof}

\end{document}